\documentclass[12pt,a4paper,fleqn]{article}

\usepackage{lscape,exscale,amsthm}

\usepackage[intlimits]{amsmath}
\usepackage{algorithm}
\usepackage{rawfonts}
\usepackage{latexsym}
\usepackage[cp850]{inputenc}
\usepackage{epsfig}
\usepackage{bbm}
\usepackage{enumerate,dsfont}
\usepackage{natbib}
\newcommand{\bep}{\mbox{\boldmath$\varepsilon$}}
\newcommand{\bmu}{\mbox{\boldmath$\mu$}}
\newcommand{\bnu}{\mbox{\boldmath$\nu$}}
\newcommand{\blambda}{\mbox{\boldmath$\lambda$}}
\newcommand{\tphi}{\tilde{\phi}}
\newcommand{\bPsi}{\mathbf{\Psi}}
\newcommand{\bx}{\mathbf{x}}

\newcommand{\bz}{\mathbf{z}}
\newcommand{\bX}{\mathbf{X}}
\newcommand{\bU}{\mathbf{U}}
\newcommand{\bF}{\mathbf{F}}

\newcommand{\bZ}{\mathbf{Z}}
\newcommand{\bH}{\mathbf{H}}
\newcommand{\bA}{\mathbf{A}}
\newcommand{\bB}{\mathbf{B}}
\newcommand{\bC}{\mathbf{C}}
\newcommand{\bS}{\mathbf{S}}
\newcommand{\bc}{\mathbf{c}}

\newcommand{\bi}{\mathbf{1}}
\newcommand{\bI}{\mathbf{I}}

\newcommand{\bV}{\mathbf{V}}

\newcommand{\E}{\mathds{E}}
\newcommand{\R}{\mathds{R}}

\newtheorem{theorem}{Theorem}
\newtheorem{lemma}{Lemma}
\newtheorem{corollary}{Corollary}

\usepackage{authblk}
\title{Objective Bayesian meta-analysis based on generalized multivariate random effects model}

\author[1]{Olha Bodnar}
\author[2]{Taras Bodnar}

\affil[1]{Unit of Statistics, School of Business, \"Orebro University, SE-70182 \"Orebro, Sweden}
\affil[2]{Department of Mathematics, Stockholm University, SE-10691 Stockholm, Sweden}

\providecommand{\keywords}[1]
{
\small	
\textbf{\textit{Keywords}:} #1
}
\begin{document}

%
\maketitle
\begin{abstract}
Objective Bayesian inference procedures are derived for the parameters of the multivariate random effects model generalized to elliptically contoured distributions. The posterior for the overall mean vector and the between-study covariance matrix is deduced by assigning two noninformative priors to the model parameter, namely the Berger and Bernardo reference prior and the Jeffreys prior, whose analytical expressions are obtained under weak distributional assumptions. It is shown that the only condition needed for the posterior to be proper is that the sample size is larger than the dimension of the data-generating model, independently of the class of elliptically contoured distributions used in the definition of the generalized multivariate random effects model. The theoretical findings of the paper are applied to real data consisting of ten studies about the effectiveness of hypertension treatment for reducing blood pressure where the treatment effects on both the systolic blood pressure and diastolic blood pressure are investigated.
\end{abstract}

\vspace{1cm}

\newpage
\section{Introduction}
Random effects model is a well established quantitative tool when the results of several studies are combined in a single values as it is usually done in meta-analysis and interlaboratory comparison studies which are widely spread in medicine, physics, chemistry, and in many other fields of science (see, e.g.,  \citet{brockwell2001, AdesHiggins2005, viechtbauer2005, Viechtbauer2007,SuttonHiggins2008, riley2010meta, StrawdermanRukhin2010, Cornell2014, novianti2014estimation, ChristianRoeverR2016, bodnar2017bayesian, rukhin2017estimation, rukhin2017research, wynants2018random, michael2019exact, veroniki2019methods}). In most of applications considered in the literature, the aim is to infer the common mean of the measurement results on a single variable, while the inference procedures for the hetorogeneity parameter have recently been derived by \citet{Rukhin2013,langan2017comparative,ma2018performance,bodnar2019} among others. Both methods of the frequentist and Bayesian statistics have been established to deal with the problem and applied in practice (see, \citet{PauleMandel1982, DersimonianLaird1986,lambert2005vague, Guolo2012, Turner2015, bodnar2017bayesian}).

Although statistical theory to analyse the univariate random effects model has been developed and successfully implemented in many applications, new challenges arise when several features are measured simultaneously and have to be combined into a single (multivariate) result. One possibility is based on the application of the univariate random effects to each feature separately. However, important information about the dependence structure present in the joint distribution of the features might be lost in this case. Another approach is to generalize the existent univariate methods to the multivariate case by deriving new statistical procedures which can capture the dependencies present between several features and efficiently combine the (multivariate) results of several studies. Moreover, the assumption of normality, which is commonly imposed in meta-analysis or in interlaboratory comparison studies, is not obviously fulfilled (see, \citet{baker2008new,lee2008flexible, BodnarLinkElster2015, jackson2018should, wang2020evaluation}) and more sophisticated statistical models which take the heavy-tailed behaviour into account should be considered in many applications. This makes an additional difficult in the practical implementation of the random effects model, since only a few observations are present in most cases and the advanced asymptotic methods cannot be longer used. For instance, \citet{Davey2011} pointed out that 75\,\% of meta-analyses reported in the Cochrane Database of Systematic Reviews (CDSR) contained five or fewer studies.

Multivariate random effects model has increased its popularity in the literature recently (see, \citet{gasparrini2012multivariate}, \citet{wei2013bayesian}, \citet{jackson2014refined}, \citet{liu2015multivariate}, \citet{noma2019efficient}, \cite{negeri2020robust}, \citet{jackson2020multivariate}). Statistical inferences for the model parameters, which are the common mean vector and the heterogeneity matrix, were initially derived from the viewpoint of the frequentist statistics. \citet{jackson2010extending} extended the DerSimonian and Laird approach to the multivariate data, while \citet{chen2012method} presented the method based on the restricted maximum likelihood approach. These two procedures from frequentist statistics constitute the commonly used methods in multivariate meta-analysis (see, e.g., \citet{jackson2013matrix}, \citet{schwarzer2015meta}, \citet{jackson2020multivariate}). \citet{paul2010bayesian} derived Bayesian inferences procedures for the parameters of the two-dimensional random effects model based on the Laplace approximation, while \citet{nam2003multivariate} provided results in a multivariate case. Both the papers discussed Bayesian inference obtained when informative priors are employed.

Following  Bernstein-von Mises theorem (see, \citet{BernardoSmith2000}), a prior has a minor impact on the posterior when the sample size is large. When a sample of a small size is available, which is a common situation in practice (cf., \citet{Davey2011}), the application of an incorrectly chosen informative prior can be very influential on the resulting Bayesian inference procedures for the model parameters. This challenge becomes even more pronounced in case of Bayesian inference for parameters of a multivariate model.

The contribution of the paper to the existent literature on multivariate random effects model and multivariate meta-analysis is done in several directions. First, we develop objective Bayesian inference procedures for the parameters of the multivariate random effects model. In particular, we derive the analytical expression of the Fisher information matrix and the two noninformative priors: Berger and Bernardo reference prior and Jeffreys prior. Employing these two priors, the expressions of the corresponding posterior distributions are obtained and the conditions for their propriety are established. Second, we weaken the assumption of multivariate normal distribution and replace it by a general class of multivariate distributions, the so-called elliptically contoured distributions (see, \citet{GuptaVargaBodnar}).

The rest of the paper is structured as follows. In Section \ref{sec:main}, the generalized multivariate random effects model is introduced and two noninformative priors, Berger and Bernardo reference prior and Jeffreys prior, are derived. The posterior distribution for model parameters are obtained in Section \ref{sec:posterior}, while the conditions for posterior propriety are stated in Section \ref{sec:propriety}. In Section \ref{sec:sampling} numerical procedures are developed to draw samples from the derived posterior distributions. Results for two special families of elliptically contoured distributions are provided in Section \ref{sec:special_cases}, while an empirical illustration is presented in Section \ref{sec:emp}. Final remarks are given in Section \ref{sec:sum}. The proofs of technical results are moved to the appendix (Section \ref{sec:app}).

\section{Model and noninformative priors}\label{sec:main}

We consider an extension of the (normal) multivariate random effects model with density function given by
\begin{equation}\label{mult-rem}
p(\bX | \bmu, \bPsi) = \frac{1}{\sqrt{\text{det}(\bPsi \otimes \bI+\bU)}} f\left(\text{vec}(\bX - \bmu \bi^\top)^\top (\bPsi \otimes \bI+\bU)^{-1}\text{vec}(\bX - \bmu \bi) \right),
\end{equation}
where $\bX$ is a $(p \times n)$ matrix, $\bmu$ is a $p$-dimensional vector, $\bPsi$ is a $(p \times p)$ matrix, $\bi$ is a vector of ones, $\bI$ is the identity matrix of an appropriate order, and $\bU$ is a $(pn \times pn)$ deterministic matrix. The symbol $\otimes$ denotes the Kronecker product, while $\text{vec}$ stands for the $vec$ operator. The model \eqref{mult-rem} extends the univariate approach suggested in \citet{BodnarLinkElster2015} to the multivariate case and can also be used when several correlated features obtained from different studies should be combined together.

In a special case of $f(z)=\exp\left(-z/2\right)/(2 \pi)^{pn/2}$ and $\bU=\text{diag}(\bU_1,...,\bU_n)$ with $\bU_i$: $p \times p$ for $i=1,...,n$, the model \eqref{mult-rem} can be written as
\begin{equation}\label{mult-rem-nor}
\bx_i=\bmu+ \blambda_i + \bep_i \quad \text{with} \quad \blambda  \sim \mathcal{N}_p(\mathbf{0},\bPsi) \quad \text{and} \quad \bep \sim \mathcal{N}_p(\mathbf{0},\bU_i) ,
\end{equation}
where $\{\blambda_i\}_{i=1,...,n}$ and $\{\bep_i\}_{i=1,...,n}$ are mutually independent. The presentation \eqref{mult-rem-nor} defines the normal multivariate random effects model. Motivated by the normal multivariate random effects model, it is assumed that $\bU=\text{diag}(\bU_1,...,\bU_n)$ holds in \eqref{mult-rem}.

In many applications in medicine, physics, and chemistry the aim is to infer $\bmu$ given observation matrix $\bX=(\bx_1,\ldots,\bx_n)$. In the applications of these fields the information about the scale matrix $\bU$ is usually provided by the participating organizations (see, \citet{lambert2005vague}, \citet{Turner2015}, \citet{bodnar2014adjustment}, \citet{jackson2020multivariate}). As a result, it is assumed to be a known symmetric positive definite matrix. On the other side, the matrix $\bPsi$ is treated as an unknown quantity with the aim to capture the additional variability in data when several observations taken at different places and times are pooled together. The matrix $\bPsi$ is usually treated as an additional nuisance parameter of the model.

By \eqref{mult-rem}, the conditional distribution of $\bX$ given $\bmu$ and $\bPsi$ belongs the class of the matrix-variate elliptical contoured distributions (see, e.g., \citet{GuptaVargaBodnar} for the definition and properties of this matrix-variate family of distributions). This assertion will be denoted by $\bX| \bmu, \bPsi \sim E_{p,n}(\bmu\bi^\top, \bPsi \otimes \bI+\bU, f)$ ($p\times n$-dimensional matrix-variate elliptically contoured distribution with location matrix $\bmu \bi^\top$, dispersion matrix $(\bPsi \otimes \bI+\bU)$, and density generator $f(.)$. Following the definition of  matrix-variate elliptically contoured distributions (see, \citet[Theorem 2.7]{GuptaVargaBodnar}), the function $f(.)$ should be a non-negative Lebesgue measurable function on $[0,\infty)$ such that
\[\int_0^\infty t^{pn-1} f(t^2)dt <\infty.\]

\subsection{Noninformative priors: Berger and Bernardo reference prior and Jeffreys prior}\label{sec:obj_priors}

In many practical applications no information or only vague information is available about the model parameters. In such cases, especially when additionally the sample size is small or the model dimension is large in comparison to the sample size, the usage of an informative prior can be questionable. As a possible solution to this problem, noninformative priors were developed and employed in the derivation of Bayesian inference. Historically, the first noninformative prior was suggested by \citet{Laplace1812} who proposed to assign a constant prior to the parameters of the model. This prior is also known in the literature as the constant prior or the uniform prior. Although the uniform prior works well when Bayesian inference is determined for location parameters of a statistical model, its application does not obviously lead to good results for other types of model parameters. One of the most crucial critiques of the uniform prior is that it is invariant under transformations of parameters.

As a solution, \citet{Jeffreys1946} proposed to compute a non-informative prior as the square root of the determinant of the Fisher information matrix. Although this approach leads to a prior which is invariant under transformations of model parameters, some difficulties arise in the case of multi-parameter statistical models (see, \citet{held2014applied}). The approach of Jeffreys was further extended in \cite{BergerBernardo1992c} who suggested the so-called reference prior (see, also \citet{BergerBernardoSun2009} for the properties of the reference prior). The idea used in the derivation of the reference prior is based on the sequential maximization of the Shannon mutual information (see, \citet{bodnar2014analytical}) which determines the distance between the prior and posterior.

In Theorem \ref{th1} the analytical expression of the Fisher information matrix is provided, which is then used in the derivation of both the Berger and Bernardo reference prior and the Jeffreys prior for the parameters of the generalized multivariate random effects model \eqref{mult-rem}.

\begin{theorem}\label{th1}
The Fisher information matrix for model (\ref{mult-rem}) with $\bU=\text{diag}(\bU_1,...,\bU_n)$ is given by
\begin{eqnarray}\label{FIM1_rem}
\bF =
	\left(
     \begin{array}{cc}
       \bF_{11} & \mathbf{O} \\
       \mathbf{O} & \bF_{22} \\
     \end{array}
   \right)
\end{eqnarray}
where
\begin{eqnarray}\label{F11}
\bF_{11}&=&\frac{4J_1}{pn}\sum_{i=1}^{n}(\bPsi + \bU_i)^{-1},\\
\bF_{22}
&=&\mathbf{G}_p^\top\Bigg[ \left( \frac{J_2}{2pn+p^2n^2} -\frac{1}{4}\right) \text{vec}\left(\sum_{i=1}^{n}(\bPsi + \bU_i)^{-1}\right)\text{vec}\left(\sum_{j=1}^{n}(\bPsi + \bU_j)^{-1}\right)^\top \nonumber\\
&+& \frac{2J_2}{2pn+p^2n^2}\sum_{i=1}^{n}\left((\bPsi + \bU_i)^{-1} \otimes (\bPsi + \bU_i)^{-1}\right)
\Bigg]\mathbf{G}_p\label{F22}
\end{eqnarray}
with
\begin{eqnarray}\label{Ji}
J_i&=&\E\left((R^2)^i\left(\frac{f^\prime\left(R^2\right)}
{f\left(R^2\right)} \right)^2 \right),
\end{eqnarray}
where $R^2=\text{vec}(\bZ)^\top\text{vec}(\bZ)$ with $\bZ \sim E_{p,n}(\mathbf{O}_{p,n},\bI_{p\times n},f)$ standard matrix-variate elliptically contoured distribution with density generator $f(.)$ and $\mathbf{G}_p$ stands for the duplication matrix.
\end{theorem}

The results of Theorem \ref{th1} show that the Fisher information matrix depends on the type of elliptical distribution only over the two univariate constants $J_1$ and $J_2$ which are fully determined by density generator $f(.)$. Moreover, the Fisher information matrix $\bF$ is finite if $J_1<\infty$ and $J_2<\infty$. Thus, it is assumed throughout the paper that the density generator $f(.)$ is chosen such that these two conditions are fulfilled. Although the expectations in the definition of $J_1$ and $J_2$ cannot always be analytically computed, they can easily be approximated via simulations by drawing samples from the corresponding standard elliptically contoured distribution. Finally, $J_1$ is present in $\bF_{11}$ as a multiplicative constant and, thus, both the Berger and Bernardo reference prior and the Jeffreys prior depend on $J_2$ only as shown below.

Since $\bF$ is block-diagonal and it does not depend on $\bmu$, the Jeffreys prior for $\bmu$ and $\bPsi$ depends on $\bPsi$ only and it is given by
\begin{equation}\label{prior_J}
\pi_J(\bmu,\bPsi)=\pi_J(\bPsi)\propto \sqrt{\text{det}(\bF)}=\sqrt{\text{det}(\bF_{11})}\sqrt{\text{det}(\bF_{22})},
\end{equation}
where $\bF_{11}$ and $\bF_{22}$ are given in \eqref{F11} and \eqref{F22}, respectively.

Moreover, using the block-diagonal structure of $\bF$ and the fact that $\bF$ does not depend on $\bmu$, we immediately obtain the Berger and Bernardo reference prior $\pi_R(\bmu, \bPsi)$ for the generalized multivariate random effects model \eqref{mult-rem} from the corollary to Proposition 5.29 in \citet{BernardoSmith2000}. This result is summarized in Theorem \ref{th2}.

\begin{theorem}\label{th2}
For model (\ref{mult-rem}) with $\bU=\text{diag}(\bU_1,...,\bU_n)$ and grouping $\{\bmu,\bPsi\}$ (i.e. with $\bPsi$ as the nuisance parameter), the Berger and Bernardo reference prior is given by
\begin{equation}\label{prior_R}
\pi_R(\bmu,\bPsi)=\pi_R(\bPsi)\propto \sqrt{\text{det}(\bF_{22})},
\end{equation}
where $\bF_{22}$ is given in \eqref{F22}.
\end{theorem}

Under additional restrictions imposed on matrix $\bU$ and density generator $f(.)$, several simplifications of the expressions of both the Jeffreys prior and the reference prior are obtained and are presented in Corollary \ref{cor1} and Corollary \ref{cor2}. For example, when the normal multivariate random effects model \eqref{mult-rem-nor} is assumed, then we get

\begin{corollary}\label{cor1}
For model (\ref{mult-rem-nor}) and grouping $\{\bmu,\bPsi\}$ (i.e. with $\bPsi$ as the nuisance parameter), the following results hold:
\begin{enumerate}[(i)]
\item the Berger and Bernardo reference prior is given by
\begin{equation}\label{prior_R_cor1}
\pi_R(\bmu,\bPsi)=\pi_R(\bPsi)\propto \sqrt{\text{det}\left(\mathbf{G}_p^\top\Bigg[\sum_{i=1}^{n}\left((\bPsi + \bU_i)^{-1} \otimes (\bPsi + \bU_i)^{-1}\right)\Bigg]\mathbf{G}_p\right)} ,
\end{equation}
\item the Jeffreys prior is given by
\begin{equation}\label{prior_J_cor1}
\pi_J(\bmu,\bPsi)=\pi_J(\bPsi)\propto \pi_R(\bPsi) \sqrt{\text{det}\left(\sum_{i=1}^{n}(\bPsi + \bU_i)^{-1}\right)} .
\end{equation}
\end{enumerate}
\end{corollary}

\begin{proof}[Proof of Corollary \ref{cor1}:]
\begin{enumerate}[(i)]
\item Using that $f(u)=\exp(-u/2)/(2\pi)^{pn/2}$, we get that $f^\prime(u)=-\frac{1}{2}f(u)$ and, consequently,
\[J_2=\frac{1}{4}\E\left((\text{vec}(\bZ)^\top\text{vec}(\bZ))^2\right)=\frac{2pn+p^2n^2}{4}\]
Hence, under model \eqref{mult-rem-nor} we obtain
\begin{eqnarray*}
\bF_{22}&=&\mathbf{G}_p^\top\Bigg[\frac{2J_2}{2pn+p^2n^2}\sum_{i=1}^{n}\left((\bPsi + \bU_i)^{-1} \otimes (\bPsi + \bU_i)^{-1}\right)\Bigg]\mathbf{G}_p,
\end{eqnarray*}
which leads to the expression presented in the statement of the corollary.

\item The result follows from the part (i) and the block-diagonality of $\bF$.
\end{enumerate}
\end{proof}

If the generalized multivariate random effects model is assumed to be homoscedastic, that is the equality $\bU_1=...=\bU_n=\bV$ holds, then the Berger and Bernardo reference prior and the Jeffreys prior are given by

\begin{corollary}\label{cor2}
Under the assumption of Theorem \ref{th2}, assume that $\bU_1=...=\bU_n=\bV$. Then
\begin{enumerate}[(i)]
\item the Berger {\rm \&} Bernardo reference prior is given by
\begin{equation}\label{prior_R_cor2}
\pi_R(\bmu,\bPsi)=\pi_R(\bPsi)\propto \text{det}\left(\bPsi + \bV\right)^{-(p+1)/2},
\end{equation}
\item the Jeffreys prior is given by
\begin{equation}\label{prior_J_cor2}
\pi_J(\bmu,\bPsi)=\pi_J(\bPsi)\propto \text{det}\left(\bPsi + \bV\right)^{-(p+2)/2}.
\end{equation}
\end{enumerate}
\end{corollary}

\begin{proof}[Proof of Corollary \ref{cor2}:]
\begin{enumerate}[(i)]
\item Under the condition $\bU_1=...=\bU_n=\bV$, we get
\begin{eqnarray*}
\bF_{22}&=&\mathbf{G}_p^\top\Bigg[ \left( \frac{J_2}{2pn+p^2n^2} -\frac{1}{4}\right) \text{vec}\left(n(\bPsi + \bV)^{-1}\right)\text{vec}\left(n(\bPsi + \bV)^{-1}\right)^\top \\
&+& \frac{2nJ_2}{2pn+p^2n^2}\left((\bPsi + \bV)^{-1} \otimes (\bPsi + \bV)^{-1}\right)\Bigg]\mathbf{G}_p.
\end{eqnarray*}

The application of the properties of determinants involving the duplication matrix $\mathbf{G}_p$ (see, \citet[Section 4.2.3]{lutkepohl1996handbook}), we obtain
\begin{eqnarray*}
&&\text{det}\Bigg\{(\mathbf{G}_p^\top\mathbf{G}_p)^{-1}\mathbf{G}_p^\top\Bigg[ \left( \frac{J_2}{2pn+p^2n^2} -\frac{1}{4}\right) \text{vec}\left(n(\bPsi + \bV)^{-1}\right)\text{vec}\left(n(\bPsi + \bV)^{-1}\right)^\top \\
&+& \frac{2nJ_2}{2pn+p^2n^2}\left((\bPsi + \bV)^{-1} \otimes (\bPsi + \bV)^{-1}\right)\Bigg]\mathbf{G}_p \Bigg\}\\
&=&\left(\frac{2nJ_2}{2pn+p^2n^2}\right)^{\frac{p(p+1)}{2}}\left(1+\left( \frac{1}{2} -\frac{2pn+p^2n^2}{8J_2}\right)pn\right)\text{det}\left(\bPsi + \bV\right)^{-(p+1)},
\end{eqnarray*}
from which the expression of the Berger and Bernardo reference prior follows.

\item The result for the Jeffreys prior follows from part (i) and the equality
\begin{eqnarray*}
\text{det}(\bF_{11})=\left(\frac{4J_1}{p}\right)^p\text{det}(\bPsi + \bV)^{-1}.
\end{eqnarray*}
\end{enumerate}
\end{proof}

It is remarkable that both the Berger and Bernardo reference prior and the Jeffreys prior under the assumption of homoscedasticity do not depend on the type of elliptically contoured distribution. In particular, the formulas from Corollary \ref{cor2} can be used for the normal multivariate random effects model \eqref{mult-rem-nor}.

\section{Posterior}\label{sec:posterior}

In the derivation of the posterior we consider a prior for $\bmu$ and $\bPsi$ which is a function of $\bPsi$ only, that is $\pi(\bPsi)$. Such a prior is an extension of both the Berger and Bernardo reference prior and the Jeffreys prior and, consequently, the derived posterior can be used to deduce the posteriors obtained when the Berger and Bernardo reference prior and the Jeffreys prior are employed as important special cases.

Under such a general prior the joint posterior for $\bmu$ and $\bPsi$ is obtained from \eqref{mult-rem} and it is given by
\begin{eqnarray}\label{post_joint}
\pi(\bmu,\bPsi|\bX) &\propto& \pi(\bPsi)p(\bX | \bmu, \bPsi) \nonumber\\
& = &\frac{\pi(\bPsi)}{\sqrt{\text{det}(\bPsi \otimes \bI+\bU)}} f\left(\text{vec}(\bX - \bmu \bi^\top)^\top (\bPsi \otimes \bI+\bU)^{-1}\text{vec}(\bX - \bmu \bi)\right),
\end{eqnarray}
with $\bU=\text{diag}(\bU_1,...,\bU_n)$. In Theorem \ref{th3} it shown that the conditional reference posterior for $\bmu$ belongs to the family of elliptically contoured distributions.

\begin{theorem}\label{th3}
Under the generalized multivariate random effects model (\ref{mult-rem}) with $\bU=\text{diag}(\bU_1,...,\bU_n)$, the conditional posterior $\pi(\bmu|\bPsi, \bx)$ is given by
\begin{eqnarray}\label{con_posterior_mu}
\pi(\bmu|\bPsi, \bX) &\propto& f_{\bPsi,\bX}\left((\bmu-\tilde{\bx}(\bPsi))^\top\left(\sum_{i=1}^{n}(\bPsi+ \bU_i)^{-1}\right)(\bmu-\tilde{\bx}(\bPsi))\right)\,,
\end{eqnarray}
where
\begin{equation}\label{g_sig-lam_bx}
f_{\bPsi,\bX}\left(u\right)=f\left(\sum_{i=1}^{n} (\bx_i-\tilde{\bx}(\bPsi))^\top (\bPsi+ \bU_i)^{-1}(\bx_i-\tilde{\bx}(\bPsi))+u\right) \qquad u \ge 0\,,
\end{equation}
with
\begin{equation}\label{tilde_x}
\tilde{\bx}(\bPsi)=\left(\sum_{i=1}^{n}(\bPsi+ \bU_i)^{-1}\right)^{-1}\sum_{i=1}^{n}(\bPsi+ \bU_i)^{-1}\bx_i.
\end{equation}
\end{theorem}

\begin{proof}[Proof of Theorem \ref{th3}:]
It holds that
\begin{eqnarray*}
&&\text{vec}(\bX - \bmu \bi^\top)^\top (\bPsi \otimes \bI+\bU)^{-1}\text{vec}(\bX - \bmu \bi)=
\sum_{i=1}^{n}(\bx_i - \bmu)^\top(\bPsi+ \bU_i)^{-1}(\bx_i - \bmu)\\
&=&(\bmu-\tilde{\bx}(\bPsi))^\top\left(\sum_{i=1}^{n}(\bPsi+ \bU_i)^{-1}\right)(\bmu-\tilde{\bx}(\bPsi))
+ \sum_{i=1}^{n} (\bx_i-\tilde{\bx}(\bPsi))^\top (\bPsi+ \bU_i)^{-1}(\bx_i-\tilde{\bx}(\bPsi))
\end{eqnarray*}
with $\tilde{\bx}(\bPsi)$ as in \eqref{tilde_x}.

Hence,
\begin{eqnarray*}
&&\pi(\bmu|\bPsi, \bX)\propto \frac{\pi(\bPsi)}{\sqrt{\text{det}(\bPsi \otimes \bI+\bU)}}\\
&\times& f\left((\bmu-\tilde{\bx}(\bPsi))^\top\left(\sum_{i=1}^{n}(\bPsi+ \bU_i)^{-1}\right)(\bmu-\tilde{\bx}(\bPsi))
+ \sum_{i=1}^{n} (\bx_i-\tilde{\bx}(\bPsi))^\top (\bPsi+ \bU_i)^{-1}(\bx_i-\tilde{\bx}(\bPsi))\right)\\
&\propto&
f_{\bPsi,\bX}\left((\bmu-\tilde{\bx}(\bPsi))^\top\left(\sum_{i=1}^{n}(\bPsi+ \bU_i)^{-1}\right)(\bmu-\tilde{\bx}(\bPsi))\right)\,,
\end{eqnarray*}
where $f_{\bPsi,\bX}(.)$ is given in (\ref{g_sig-lam_bx}).
\end{proof}

As a straightforward consequence of the result in Theorem \ref{th3} by substituting $\pi(\bPsi)$ with $\pi_R(\bPsi)$ and $\pi_J(\bPsi)$ from \eqref{prior_R} and \eqref{prior_J} respectively, we get the conditional reference posterior $\pi(\bmu|\bPsi, \bx)$ for the generalized multivariate random effects model (\ref{mult-rem}) and the conditional posterior when the Jeffreys prior is used.

Moreover, from the proof of Theorem \ref{th3} we also get the marginal posterior for $\bPsi$ as given by

\begin{corollary}\label{cor3}
Under the generalized multivariate random effects model (\ref{mult-rem}) with $\bU=\text{diag}(\bU_1,...,\bU_n)$, the conditional posterior $\pi(\bPsi| \bX)$ is given by
\begin{eqnarray}\label{marg_posterior_Psi}
\pi(\bPsi| \bX) &\propto& \frac{\pi(\bPsi)}{\sqrt{\text{det}(\sum_{i=1}^{n}(\bPsi + \bU_i)^{-1})}\prod_{i=1}^{n}\sqrt{\text{det}(\bPsi + \bU_i)}}\nonumber\\
&\times&\int_{0}^{\infty} u^{p-1} f\left(u^2+\sum_{i=1}^{n} (\bx_i-\tilde{\bx}(\bPsi))^\top (\bPsi+ \bU_i)^{-1}(\bx_i-\tilde{\bx}(\bPsi))\right) \mathbf{d}u\,,
\end{eqnarray}
\end{corollary}

\begin{proof}[Proof of Corollary \ref{cor3}:]
Using the transformation $\bnu=\left(\sum_{i=1}^{n}\bPsi +\bU_i\right)^{-1}(\bmu-\tilde{\bx}(\bPsi))$ with the Jacobian $1/\sqrt{\text{det}(\sum_{i=1}^{n}(\bPsi + \bU_i)^{-1})}$, the marginal posterior for $\bPsi$ is expressed as
\begin{eqnarray*}
\pi(\bPsi| \bX) &\propto& \int_{\R^p}  \frac{\pi(\bPsi)}{\prod_{i=1}^{n}\sqrt{\text{det}(\bPsi + \bU_i)}} f\left(\text{vec}(\bX - \bmu \bi^\top)^\top (\bPsi \otimes \bI+\bU)^{-1}\text{vec}(\bX - \bmu \bi)\right)\mathbf{d} \bmu\\
&=&\frac{\pi(\bPsi)}{\sqrt{\text{det}(\sum_{i=1}^{n}(\bPsi + \bU_i)^{-1})}\prod_{i=1}^{n}\sqrt{\text{det}(\bPsi +\bU_i)}}\\
&\times& \int_{\R^p}f\left(\bnu^\top \bnu+\sum_{i=1}^{n} (\bx_i-\tilde{\bx}(\bPsi))^\top (\bPsi+ \bU_i)^{-1}(\bx_i-\tilde{\bx}(\bPsi))\right)\mathbf{d} \bnu.
\end{eqnarray*}

The rest of the proof follows by noting that (cf., \citet[Lemma 2.1]{GuptaVargaBodnar}):
\[\int_{\R^p} g(\mathbf{t}^\top \mathbf{t}) \mathbf{d} \mathbf{t}= \frac{2\pi^{p/2}}{\Gamma(p/2)} \int_{0}^{\infty} u^{p-1} g(u^2)du,\]
where $\Gamma(.)$ stands for the gamma function.
\end{proof}

The posterior mean vector and the posterior covariance matrix of $\bmu$ are derived from Theorem \ref{th3} by using the rule of iterated expectations. They are given by
\begin{eqnarray}\label{mu_post_mean}
\E\left(\bmu|\bX\right)&=&\E(\E(\bmu|\bPsi,\bX)|\bX)=\E(\tilde{\bx}(\bPsi)|\bX) \nonumber\\
&=&\E\left(\left(\sum_{i=1}^{n}(\bPsi+ \bU_i)^{-1}\right)^{-1}\sum_{i=1}^{n}(\bPsi+ \bU_i)^{-1}\bx_i\Bigg|\bX\right)
\end{eqnarray}
and
\begin{eqnarray}\label{mu_post_var}
\mathds{V}ar\left(\bmu|\bX\right)&=&\E(\mathds{V}ar(\bmu|\bPsi,\bX)|\bX)+\mathds{V}ar(\E(\bmu|\bPsi,\bX)|\bX) \nonumber\\
&=&\E\left(C(\bPsi)\left(\sum_{i=1}^{n}(\bPsi+ \bU_i)^{-1}\right)^{-1}\Bigg|\bX\right)\nonumber\\
&+&\mathds{V}ar\left(\left(\sum_{i=1}^{n}(\bPsi+ \bU_i)^{-1}\right)^{-1}\sum_{i=1}^{n}(\bPsi+ \bU_i)^{-1}\bx_i\Bigg|\bX\right),
\end{eqnarray}
where
\begin{equation}\label{C_bPsi}
C(\bPsi)=\E\left(R_{\bPsi,\bX}^2\right)
\quad \text{with} \quad
R_{\bPsi,\bX}^2=\mathbf{z}_{\bPsi,\bX}^\top \mathbf{z}_{\bPsi,\bX},
\end{equation}
where $\mathbf{z}_{\bPsi,\bX}\sim E_p(\mathbf{0}_p,\mathbf{I}_p,f_{\bPsi,\bX})$. 

\subsection{Propriety}\label{sec:propriety}

For the derivation of the conditions required for the propriety of the posterior $\pi(\bmu,\bPsi|\bX)$, we use the following lemma:

\begin{lemma}\label{lem1}
Let $\bA>0$ be a symmetric and positive definite matrix, and let $\bB \ge 0$ be a symmetric and positive semidefinite matrix. Then
$$\bA^{-1} \otimes \bA^{-1} -(\bA+\bB)^{-1} \otimes (\bA+\bB)^{-1}\ge 0$$
is positive semidefinite.
\end{lemma}

The proof of Lemma~\ref{lem1} is given in the appendix. In Theorem~\ref{th4} we formulate the conditions required for the propriety of the posterior.
\begin{theorem}\label{th4}
Consider the generalized multivariate random effects model (\ref{mult-rem}) with $\bU=\text{diag}(\bU_1,...,\bU_n)$. Let $f(u)$ be a non-increasing function in $u \ge 0$ and $\frac{J_2}{2pn+p^2n^2} -\frac{1}{4} \le 0$ where $J_2$ is defined in \eqref{Ji}.
\begin{enumerate}
\item If $n \ge p$, then the posterior $\pi(\bmu,\bPsi| \bX)$ derived under the Jeffreys prior $\pi_J(\bPsi)$ is proper.
\item If $n\ge p+1$, then the posterior $\pi(\bmu,\bPsi| \bX)$ derived under the Berger and Bernardo reference prior $\pi_R(\bPsi)$ is proper.
\end{enumerate}

\end{theorem}

\begin{proof}[Proof of Theorem \ref{th4}:]
First, we derive an upper bound for the determinant of two diagonal blocks of the Fisher information matrix derived in Theorem~\ref{th1}. Let $\tphi=\min_{i=1,...,n}\min_{j=1,...,p} \phi_{i,j}$ denote the minimum of the eigenvalues $\{\phi_{i,j}\}_{j=1,...,p}$ computed for the matrices $\bU_i$, $i=1,...,n$. Since $\bPsi+\tphi\bI$ and $\bU_i-\tphi\bI$ are symmetric and positive semi-definite
%
%
we obtain from Lemma~\ref{lem1}.(ii) that
\[(\bPsi+\tphi\bI)^{-1}\otimes (\bPsi+\tphi\bI)^{-1} - (\bPsi+\bU_i)^{-1}\otimes (\bPsi+\bU_i)^{-1} \ge \mathbf{0},\]
for $i=1,...,n$ and, hence,
\[n(\bPsi+\tphi\bI)^{-1}\otimes (\bPsi+\tphi\bI)^{-1} - \sum_{i=1}^n(\bPsi+\bU_i)^{-1}\otimes (\bPsi+\bU_i)^{-1} \ge \mathbf{0}.\]
The inequality ${J_2}/{(2pn+p^2n^2)} -{1}/{4}\le 0$ yields
\begin{eqnarray*}
\text{det}(\bF_{22})&\le& \text{det}\Bigg(\frac{2J_2}{2pn+p^2n^2} \mathbf{G}_p^\top\sum_{i=1}^{n}\left((\bPsi + \bU_i)^{-1} \otimes (\bPsi + \bU_i)^{-1}\right)\mathbf{G}_p
\Bigg),
\end{eqnarray*}
while the application of Theorem 18.1.6 of \citet{Harville97} implies
\begin{eqnarray}
\text{det}(\bF_{22})&\le& \text{det}\Bigg(\frac{2nJ_2}{2pn+p^2n^2} \mathbf{G}_p^\top\left((\bPsi+\tphi\bI)^{-1} \otimes (\bPsi+\tphi\bI)^{-1}\right)\mathbf{G}_p
\Bigg) \nonumber \\
&=& \left(\frac{2J_2}{2p+p^2n} \right)^{p(p+1)/2} \text{det}(\bPsi+\tphi\bI)^{-(p+1)}
\label{approx_F22}.
\end{eqnarray}

Let
$$
c_{11}(p,n)=\left(\frac{4J_1}{p}\right)^{p/2} \quad \text{and} \quad
c_{22}(p,n)=\left(\frac{2J_2}{2p+p^2n} \right)^{p(p+1)/2}.
$$
Using that
\begin{equation}\label{det_prod}
\left(\text{det}(\bPsi \otimes \bI+\bU)\right)^{-1}=\prod_{i=1}^n\text{det}\left((\bPsi +\bU_i)^{-1}\right)
\le \text{det}\left((\bPsi +\tphi\bI)^{-1}\right)^n=\text{det}(\bPsi +\tphi\bI)^{-n}
\end{equation}
and that $f(.)$ is non-increasing, we get that the kernel of the posterior $\pi(\bmu,\bPsi| \bX)$ derived under the Jeffreys prior is bounded by
{\small
\begin{eqnarray}
&&\hspace{-1cm}g_J(\bmu,\bPsi| \bX) = \frac{\pi_J(\bPsi)}{\sqrt{\text{det}(\bPsi \otimes \bI+\bU)}} f\left(\text{vec}(\bX - \bmu \bi^\top)^\top (\bPsi \otimes \bI+\bU)^{-1}\text{vec}(\bX - \bmu \bi)\right) \nonumber\\
&=& c_{11}(p,n) c_{22}(p,n)\frac{\sqrt{\text{det}(\bF_{22})}}{\sqrt{\text{det}(\bPsi \otimes \bI+\bU)}} \sqrt{\text{det}\left(\sum_{i=1}^{n}(\bPsi+ \bU_i)^{-1}\right)}
\nonumber\\
&\times& f\left((\bmu-\tilde{\bx}(\bPsi))^\top\left(\sum_{i=1}^{n}(\bPsi+ \bU_i)^{-1}\right)(\bmu-\tilde{\bx}(\bPsi))
+ \sum_{i=1}^{n} (\bx_i-\tilde{\bx}(\bPsi))^\top (\bPsi+ \bU_i)^{-1}(\bx_i-\tilde{\bx}(\bPsi))\right)\nonumber\\
&\le& c_{11}(p,n) c_{22}(p,n)\text{det}(\bPsi+\tphi\bI)^{-(n+p+1)/2}\label{g_J1}
\nonumber\\
&\times&\sqrt{\text{det}\left(\sum_{i=1}^{n}(\bPsi+ \bU_i)^{-1}\right)}f\left((\bmu-\tilde{\bx}(\bPsi))^\top\left(\sum_{i=1}^{n}(\bPsi+ \bU_i)^{-1}\right)(\bmu-\tilde{\bx}(\bPsi))
\right),\label{g_J2}
\end{eqnarray}
}
where \eqref{g_J1} is proportional to the kernel of a generalized matrix-variate beta type II distribution with parameters $(p+1)/2$, $n/2$, and $\tphi\bI$ (see, e.g., \citet{gupta2000matrix}), which is a proper density as soon as $p \le n$. The expression \eqref{g_J2} is the kernel of the $p$-dimensional elliptically countered distribution with density generator $f(.)$. Since $f(.)$ is a density generator of a $p\times n$-dimensional matrix-variate elliptically contoured distribution, the integral of \eqref{g_J2} over $\bmu \in \mathds{R}^p$ converges. Hence, the joint posterior $\pi(\bmu,\bPsi| \bX)$ derived under the Jeffreys prior is proper under the condition $p \le n$.

In the case of the Berger and Bernardo reference prior we get that the posterior for $\bmu$ and $\bPsi$ is bounded by
\begin{eqnarray*}
&&\hspace{-1cm}g_R(\bmu,\bPsi| \bX) =\frac{\pi_R(\bPsi)}{\sqrt{\text{det}(\bPsi \otimes \bI+\bU)}} f\left(\text{vec}(\bX - \bmu \bi^\top)^\top (\bPsi \otimes \bI+\bU)^{-1}\text{vec}(\bX - \bmu \bi)\right) \nonumber\\
&=& c_{22}(p,n) \frac{\sqrt{\text{det}(\bF_{22})}}{\sqrt{\text{det}(\bPsi \otimes \bI+\bU)}}
\frac{1}{\sqrt{\text{det}\left(\sum_{i=1}^{n}(\bPsi+ \bU_i)^{-1}\right)}}
\nonumber\\
&\times&\sqrt{\text{det}\left(\sum_{i=1}^{n}(\bPsi+ \bU_i)^{-1}\right)}f\left((\bmu-\tilde{\bx}(\bPsi))^\top\left(\sum_{i=1}^{n}(\bPsi+ \bU_i)^{-1}\right)(\bmu-\tilde{\bx}(\bPsi))
\right).
\end{eqnarray*}

It holds that (see, \citet[p.55]{lutkepohl1996handbook})
\begin{eqnarray*}
n^p\text{det}\left(\frac{1}{n}\sum_{i=1}^{n}(\bPsi+ \bU_i)^{-1}\right)\ge n^p\prod_{i=1}^n \text{det}\left((\bPsi+ \bU_i)^{-1}\right)^{1/n}=n^p \prod_{i=1}^n \text{det}(\bPsi+ \bU_i)^{-1/n},
\end{eqnarray*}
which together with \eqref{det_prod} implies that
\begin{eqnarray}
&&g_R(\bmu,\bPsi| \bX) \le n^{-p/2} c_{22}(p,n) \frac{\sqrt{\text{det}(\bF_{22})}}{\sqrt{\text{det}(\bPsi \otimes \bI+\bU)^{1-1/n}}}
\nonumber\\
&\times&\sqrt{\text{det}\left(\sum_{i=1}^{n}(\bPsi+ \bU_i)^{-1}\right)}f\left((\bmu-\tilde{\bx}(\bPsi))^\top\left(\sum_{i=1}^{n}(\bPsi+ \bU_i)^{-1}\right)(\bmu-\tilde{\bx}(\bPsi))
\right)\nonumber\\
&\le&n^{-p/2} c_{22}(p,n)
\text{det}(\bPsi+\tphi\bI)^{-(n+p)/2}\label{g_R1}\\
&\times&\sqrt{\text{det}\left(\sum_{i=1}^{n}(\bPsi+ \bU_i)^{-1}\right)}f\left((\bmu-\tilde{\bx}(\bPsi))^\top\left(\sum_{i=1}^{n}(\bPsi+ \bU_i)^{-1}\right)(\bmu-\tilde{\bx}(\bPsi))
\right).\label{g_R2}
\end{eqnarray}
The last line coincides with \eqref{g_J2} and it is integrable in $\mathds{R}^p$. Moreover, \eqref{g_R1} is proportional to the kernel of a generalized matrix-variate beta type II distribution with parameters $(p+1)/2$, $(n-1)/2$, and $\tphi\bI$ (see, e.g., \citet{gupta2000matrix}), which is a proper density as soon as $p+1 \le n$. Thus, the joint posterior $\pi(\bmu,\bPsi| \bX)$ derived under the Berger and Bernardo reference prior is proper under the condition $p+1 \le n$.
\end{proof}

\section{Drawing samples from the posterior distribution:\\ Metropolis-Hastings algorithm}\label{sec:sampling}

In this section we develop algorithms to draw samples $(\bmu^{(b)}, \bPsi^{(b)})$ from the posterior derived under the Berger and Bernardo reference prior and the Jeffreys prior. The idea is based on the application of the Markov chain Monte Carlo based on the Metropolis-Hastings algorithm, a popular approach is Bayesian statistics (see, e.g., \citet{givens2012computational}). Recently, \citet{hill2019stationarity} provided a comprehensive discussion of the stationarity and convergence of the algorithm, that depends on the chosen proposal from which the samples are generated. A good proposal distribution should have the support which covers the support of the target distribution, i.e., of the posterior for $\bmu$ and $\bPsi$. Also, it should ensure that that the constructed Markov chain has good mixing properties and it will not stack in a single point.

As a proposal, we suggest to use the special case of the posterior distribution derived under each of the considered prior in the case $\bU_1=...=\bU_n=\mathbf{O}$. The two proposals are then defined for all positive semi-definite matrices, thus having the same supports as the two posteriors derived under the Berger and Bernardo reference prior and the Jeffreys prior. More precisely, ignoring the normalizing constants the proposal under the Berger and Bernardo reference prior is given by
\begin{eqnarray*}
q_R(\bmu,\bPsi| \bX) &=& \text{det}(\bPsi)^{-(n+p+1)/2} f\left(\text{tr}\left(\bPsi^{-1}\sum_{i=1}^n(\bx_i - \bmu)(\bx_i - \bmu)^\top \right)\right),
\end{eqnarray*}
and it is expressed as
\begin{eqnarray*}
q_J(\bmu,\bPsi| \bX) &=& \text{det}(\bPsi)^{-(n+p+2)/2} f\left(\text{tr}\left(\bPsi^{-1}\sum_{i=1}^n(\bx_i - \bmu)(\bx_i - \bmu)^\top \right)\right)
\end{eqnarray*}
under the Jeffreys prior.

Let
\begin{equation}\label{barx_bS}
\bar{\bx}=\frac{1}{n} \sum_{i=1}^n \bx_i
\quad \text{and} \quad
\bS=\frac{1}{n-1}\sum_{i=1}^n(\bx_i - \bar{\bx})(\bx_i - \bar{\bx})^\top,
\end{equation}
In using that
\begin{eqnarray}\label{sum-th4}
&& \sum_{i=1}^n(\bx_i - \bmu)(\bx_i - \bmu)^\top = (n-1)\bS +n(\bmu-\bar{\bx})(\bmu-\bar{\bx})^\top,
\end{eqnarray}
which implies
\[
\text{det}(\sum_{i=1}^n(\bx_i - \bmu)(\bx_i - \bmu)^\top)=\text{det}((n-1)\bS)\left(1 +\frac{n}{n-1}(\bmu-\bar{\bx})^\top\bS^{-1}(\bmu-\bar{\bx}) \right)
\]
we get
\begin{eqnarray}\label{q_R}
&&q_R(\bmu,\bPsi| \bX) \propto \left(1 +\frac{1}{n-p}\frac{n(n-p)}{n-1}(\bmu-\bar{\bx})^\top\bS^{-1}(\bmu-\bar{\bx}) \right)^{-n/2} \\
&\times&\text{det}(\bPsi)^{-(n+p+1)/2} \text{det}(\sum_{i=1}^n(\bx_i - \bmu)(\bx_i - \bmu)^\top )^{n/2}
f\left(\text{tr}\left(\bPsi^{-1}\sum_{i=1}^n(\bx_i - \bmu)(\bx_i - \bmu)^\top \right)\right) \nonumber
\end{eqnarray}
and
\begin{eqnarray}\label{q_J}
&&q_J(\bmu,\bPsi| \bX) \propto \left(1 +\frac{1}{n-p+1}\frac{n(n-p+1)}{n-1}(\bmu-\bar{\bx})^\top\bS^{-1}(\bmu-\bar{\bx}) \right)^{-(n+1)/2} \\
&\times&\text{det}(\bPsi)^{-(n+p+2)/2} \text{det}(\sum_{i=1}^n(\bx_i - \bmu)(\bx_i - \bmu)^\top )^{(n+1)/2}
f\left(\text{tr}\left(\bPsi^{-1}\sum_{i=1}^n(\bx_i - \bmu)(\bx_i - \bmu)^\top \right)\right) .\nonumber
\end{eqnarray}

The expression of the proposal $q_R(\bmu,\bPsi| \bX)$ derived under the Berger and Bernardo reference prior is proportional to the joint density function of $\bmu$ and $\bPsi$ with $\bPsi|\bmu, \bX \sim GIW_{p}(n+p+1,\sum_{i=1}^n(\bx_i - \bmu)(\bx_i - \bmu)^\top,f)$ (generalized $p$-dimensional inverse Wishart distribution with $n+p+1$ degrees of freedom, scale matrix $\sum_{i=1}^n(\bx_i - \bmu)(\bx_i - \bmu)^\top$, and density generator $f$, see, e.g., \citet{sutradhar1989generalization}) and $\bmu|\bX \sim t_p  \left(n-p,\bar{\bx},\dfrac{(n-1)\bS} {n(n-p)}\right)$ ($p$-dimensional multivariate $t$-distribution with $n-p$ degrees of freedom, location vector $\bar{\bx}$, and scale matrix $\dfrac{(n-1)\bS} {n(n-p)}$). Similarly, we get that the proposal under the Jeffreys prior is proportional to the joint density function of $\bmu$ and $\bPsi$ with $\bPsi|\bmu, \bX \sim GIW_{p}(n+p+2,\sum_{i=1}^n(\bx_i - \bmu)(\bx_i - \bmu)^\top,f)$ and $\bmu|\bX \sim t_p  \left(n-p+1,\bar{\bx},\dfrac{(n-1)\bS} {n(n-p+1)}\right)$.

We finally note that both proposals \eqref{q_R} and \eqref{q_J} are proper under the conditions $n\ge p+1$ and $n \ge p$, respectively, which coincides with the conditions needed for the propriety of the posteriors derived under the Berger and Bernardo reference prior and the Jeffreys prior in Theorem~\ref{th4}. As a result, the suggested proposal possesses the similar tail behaviour as the corresponding posteriors and, thus, they are good candidates for the construction of the Markov chains.

\begin{algorithm}
\caption{Metropolis-Hastings algorithm for drawing realizations from $\pi(\bmu, \bPsi|\bX)$ as in \eqref{post_joint} under the Berger and Bernardo reference prior \eqref{prior_R}}\label{algorithm_R_ell}
\begin{enumerate}[(1)]
\item \textbf{Initialization:} Choose the initial values $\bmu^{(0)}$ and $\bPsi^{(0)}$ for $\bmu$ and $\bPsi$ and set $b =0$.
\item \textbf{Generating new values of $\bmu^{(w)}$ and $\bPsi^{(w)}$ from the proposal:}
\begin{enumerate}[(i)]
\item For given data $\bX=(\bx_1,...,\bx_n)$, generate $\bmu^{(w)}$ from $t_p\left(n-p,\bar{\bx},\dfrac{(n-1)\bS} {n(n-p)}\right)$ with $\bar{\bx}$ and $\bS$ as in \eqref{barx_bS};
\item Using data $\bX$ and the drawn in step (i) $\bmu^{(w)}$, generate $\bPsi^{(w)}$ from $\bPsi|\bmu=\bmu^{(w)}, \bX \sim GIW_{p}(n+p+1,\sum_{i=1}^n(\bx_i - \bmu^{(w)})(\bx_i - \bmu^{(w)})^\top,f)$.
\end{enumerate}
\item \textbf{Computation of the Metropolis-Hastings ratio:}
\[MH^{(b)}=\frac{\pi(\bmu^{(w)}, \bPsi^{(w)}|\bX)q_R(\bmu^{(b-1)},\bPsi^{(b-1)}| \bX)}{\pi(\bmu^{(b-1)}, \bPsi^{(b-1)}|\bX)q_R(\bmu^{(w)},\bPsi^{(w)}| \bX)}.\]
\item \textbf{Moving to the next state of the Markov chain:}
\begin{enumerate}[(i)]
\item Generate $U^{(b)}$ from the uniform distribution on $[0,1]$;
\item If $U^{b}<\min\left\{1,MH^{(b)}\right\}\pi(\bmu^{(b)}$, then set $\bmu^{(b)}=\bmu^{(w)}$ and $\bPsi^{(b)}=\bPsi^{(w)}$ (Markov chain moves to the new state). Otherwise, set $\bmu^{(b)}=\bmu^{(b-1)}$ and $\bPsi^{(b)}=\bPsi^{(b-1)}$ (Markov chain stays in the previous state).
\end{enumerate}
\item Return to step (2), increase $b$ by 1, and repeat until the sample of size $B$ is accumulated.
\end{enumerate}
\end{algorithm}

The Metropolis-Hastings algorithm for generating a draw from $\pi(\bmu, \bPsi|\bX)$ derived under the Berger and Bernardo reference prior is given in Algorithm~\ref{algorithm_R_ell}.
A similar algorithm with minor changes is constructed to draw a sample from the posterior $\pi(\bmu, \bPsi|\bX)$ derived under the Jeffreys prior. It is summarized in Algorithm~\ref{algorithm_J_ell}.

\begin{algorithm}
\caption{Metropolis-Hastings algorithm for drawing realizations from $\pi(\bmu, \bPsi|\bX)$ as in \eqref{post_joint} under the Jeffreys prior \eqref{prior_J}}\label{algorithm_J_ell}
\begin{enumerate}[(1)]
\item \textbf{Initialization:} Choose the initial values $\bmu^{(0)}$ and $\bPsi^{(0)}$ for $\bmu$ and $\bPsi$ and set $b =0$.
\item \textbf{Generating new values of $\bmu^{(w)}$ and $\bPsi^{(w)}$ from the proposal:}
\begin{enumerate}[(i)]
\item For given data $\bX=(\bx_1,...,\bx_n)$, generate $\bmu^{(w)}$ from $t_p\left(n-p+1,\bar{\bx},\dfrac{(n-1)\bS} {n(n-p+1)}\right)$ with $\bar{\bx}$ and $\bS$ as in \eqref{barx_bS};
\item Using data $\bX$ and the drawn in step (i) $\bmu^{(w)}$, generate $\bPsi^{(w)}$ from $\bPsi|\bmu=\bmu^{(w)}, \bX \sim GIW_{p}(n+p+2,\sum_{i=1}^n(\bx_i - \bmu^{(w)})(\bx_i - \bmu^{(w)})^\top,f)$.
\end{enumerate}
\item \textbf{Computation of the Metropolis-Hastings ratio:}
\[MH^{(b)}=\frac{\pi(\bmu^{(w)}, \bPsi^{(w)}|\bX)q_J(\bmu^{(b-1)},\bPsi^{(b-1)}| \bX)}{\pi(\bmu^{(b-1)}, \bPsi^{(b-1)}|\bX)q_J(\bmu^{(w)},\bPsi^{(w)}| \bX)}.\]
\item \textbf{Moving to the next state of the Markov chain:}
\begin{enumerate}[(i)]
\item Generate $U^{(b)}$ from the uniform distribution on $[0,1]$;
\item If $U^{b}<\min\left\{1,MH^{(b)}\right\}\pi(\bmu^{(b)}$, then set $\bmu^{(b)}=\bmu^{(w)}$ and $\bPsi^{(b)}=\bPsi^{(w)}$ (Markov chain moves to the new state). Otherwise, set $\bmu^{(b)}=\bmu^{(b-1)}$ and $\bPsi^{(b)}=\bPsi^{(b-1)}$ (Markov chain stays in the previous state).
\end{enumerate}
\item Return to step (2), increase $b$ by 1, and repeat until the sample of size $B$ is accumulated.
\end{enumerate}
\end{algorithm}

\section{Several families of elliptical distributions}\label{sec:special_cases}

In this section we apply the obtained theoretical results in case of two special families of elliptically contoured distribution: normal distribution and $t$-distribution.

\subsection{Normal multivariate random effects model}\label{sec:normal}
In the case of the normal multivariate random effects model \eqref{mult-rem-nor}, we have
\begin{equation}\label{f_nor}
f(u)=K_{p,n}\exp(-u/2)~~ \text{with} ~~ K_{p,n}=(2\pi)^{-pn/2},
\end{equation}
which directly yields
\begin{eqnarray*}
f_{\bPsi,\bX}\left(u\right)=\frac{1}{(2\pi)^{pn/2}} \exp\left(-\frac{u}{2}\right)\exp\left(-\frac{1}{2}\sum_{i=1}^{n} (\bx_i-\tilde{\bx}(\bPsi))^\top (\bPsi+ \bU_i)^{-1}(\bx_i-\tilde{\bx}(\bPsi))\right).
\end{eqnarray*}

The last equality leads to the conclusion that the conditional posterior for $\bmu$ given $\bPsi$ is a multivariate normal distribution expressed as
\begin{eqnarray}\label{con_posterior_mu_nor}
\bmu|\bPsi,\bX \sim \mathcal{N}\left(\left(\sum_{i=1}^{n}(\bPsi+ \bU_i)^{-1}\right)^{-1}\sum_{i=1}^{n}(\bPsi+ \bU_i)^{-1}\bx_i,\left(\sum_{i=1}^{n}(\bPsi+ \bU_i)^{-1}\right)^{-1}\right),
\end{eqnarray}
while the marginal posterior for $\bPsi$ is given by
\begin{eqnarray}\label{marg_posterior_Psi_nor}
\pi(\bPsi|\bX)&\propto& \frac{\pi(\bPsi)}{\sqrt{\text{det}(\sum_{i=1}^{n}(\bPsi + \bU_i)^{-1})}\prod_{i=1}^{n}\sqrt{\text{det}(\bPsi + \bU_i)}}\nonumber\\
&\times& \exp\left(-\frac{1}{2}\sum_{i=1}^{n} (\bx_i-\tilde{\bx}(\bPsi))^\top (\bPsi+ \bU_i)^{-1}(\bx_i-\tilde{\bx}(\bPsi))\right)\,.
\end{eqnarray}
The posterior mean vector and the posterior covariance matrix of $\bmu$ are obtained as in \eqref{mu_post_mean} and \eqref{mu_post_var} with $C(\bPsi)=1$. Finally, we note that the posterior $\pi(\bmu,\bPsi|\bX)$ is proper for $n\ge p+1$ for the Berger and Bernardo reference prior and for $n \ge p$ for the Jeffreys prior following Theorem~\ref{th4}, since $\exp(-u/2)$ is a decreasing function in $u$ and $\frac{J_2}{2pn+p^2n^2}-\frac{1}{4}=0$.

All the derived expressions for the normal multivariate random effects model, like conditional posterior for $\bmu$, posterior mean vector, etc., depend on the marginal posterior for $\bPsi$ and thus cannot be computed analytically. In the univariate case, \citet{BodnarLinkElster2015} suggested a numerical procedure for the computation of such quantities based on the evaluation of one-dimensional integral. In the multivariate case $\bPsi$ is a matrix now and since it should be positive semidefinite it imposes further complications on the numerical integration. For that reason we opt for the simulation-based approach as described in Section~\ref{sec:sampling}.

For generating samples from the posterior $\pi(\bmu,\bPsi|\bX)$ we apply Algorithm 1 under the Berger and Bernardo reference prior and Algorithm 2 under the Jeffreys prior where the inverse generalized Wishart distribution becomes the inverse Wishart distribution with $n+p+1$ and $n+p+2$ degrees of freedom, respectively. Other parts of the algorithms remain the same without changes.

Alternatively, one can modify these two algorithms using the properties of the normal distribution. Under the Berger and Bernardo reference prior, another proposal distribution can be constructed by using \eqref{q_R} with \eqref{f_nor}. Namely, from \eqref{barx_bS} and \eqref{sum-th4} we get
\begin{eqnarray*}
&&\exp\left(-\frac{1}{2}\text{tr}\left(\bPsi^{-1}\sum_{i=1}^n(\bx_i - \bmu)(\bx_i - \bmu)^\top \right)\right)\\
&=&\exp\left(-\frac{n-1}{2}\text{tr}\left(\bPsi^{-1}\bS\right)\right) \exp\left(-\frac{n}{2}(\bmu-\bar{\bx})^\top\bPsi^{-1}(\bmu-\bar{\bx})\right).
\end{eqnarray*}
This leads to the algorithm derived under the Berger and Bernardo reference prior which is summarized in Algorithm 3. A similar approach can also be used when the Jeffreys prior is employed with the only change in step (2) of Algorithm 3, where $\bPsi^{(w)}|\bX \sim IW_{p}(n+p,(n-1)\bS)$ should be replaced by $\bPsi^{(w)}|\bX \sim IW_{p}(n+p+1,(n-1)\bS)$.

\begin{algorithm}
\caption{Metropolis-Hastings algorithm for drawing realizations from $\pi(\bmu, \bPsi|\bX)$ as in \eqref{post_joint} under the Berger and Bernardo reference prior \eqref{prior_R} in the normal multivariate random effects model}\label{algorithm_R_nor}
\begin{enumerate}[(1)]
\item \textbf{Initialization:} Choose the initial values $\bmu^{(0)}$ and $\bPsi^{(0)}$ for $\bmu$ and $\bPsi$ and set $b =0$.
\item \textbf{Generating new values of $\bmu^{(w)}$ and $\bPsi^{(w)}$ from the proposal:}
\begin{enumerate}[(i)]
\item For given data $\bX=(\bx_1,...,\bx_n)$, generate $\bPsi^{(w)}$ from $\bPsi^{(w)}|\bX \sim IW_{p}(n+p,(n-1)\bS)$;
\item using data $\bX$ and the drawn in step (i) $\bPsi^{(w)}$, generate $\bmu^{(w)}$ from $\bmu|\bPsi=\bPsi^{(w)}, \bX \sim N_p\left(\bar{\bx},\dfrac{\bPsi^{(w)}} {n}\right)$ with $\bar{\bx}$ and $\bS$ as in \eqref{barx_bS}.
\end{enumerate}
\item \textbf{Computation of the Metropolis-Hastings ratio:}
\[MH^{(b)}=\frac{\pi(\bmu^{(w)}, \bPsi^{(w)}|\bX)q_R(\bmu^{(b-1)},\bPsi^{(b-1)}| \bX)}{\pi(\bmu^{(b-1)}, \bPsi^{(b-1)}|\bX)q_R(\bmu^{(w)},\bPsi^{(w)}| \bX)}.\]
\item \textbf{Moving to the next state of the Markov chain:}
\begin{enumerate}[(i)]
\item Generate $U^{(b)}$ from the uniform distribution on $[0,1]$;
\item If $U^{b}<\min\left\{1,MH^{(b)}\right\}\pi(\bmu^{(b)}$, then set $\bmu^{(b)}=\bmu^{(w)}$ and $\bPsi^{(b)}=\bPsi^{(w)}$ (Markov chain moves to the new state). Otherwise, set $\bmu^{(b)}=\bmu^{(b-1)}$ and $\bPsi^{(b)}=\bPsi^{(b-1)}$ (Markov chain stays in the previous state).
\end{enumerate}
\item Return to step (2), increase $b$ by 1, and repeat until the sample of size $B$ is accumulated.
\end{enumerate}
\end{algorithm}

The performance of two algorithms for drawing samples from the posterior distribution is studied in Figures \ref{fig:nor-mu1} for the normal multivariate random effects model when the Berger and Bernardo reference prior and the Jeffreys prior are employed. The notation 'Algorithm A' corresponds to the case where $\bmu$ is drawn from the marginal distribution and $\bPsi$ is generated from the conditional distribution as in Algorithms \ref{algorithm_R_ell} and \ref{algorithm_J_ell} with density generator $f(.)$ as in \eqref{f_nor}, while the notation 'Algorithm B' corresponds to the case when $\bPsi$ is generated from the marginal distribution and $\bmu$ is obtained from the conditional distribution as in Algorithm \ref{algorithm_R_nor}.

\begin{figure}[h!t]
\centering
\begin{tabular}{cc}
\includegraphics[width=8cm]{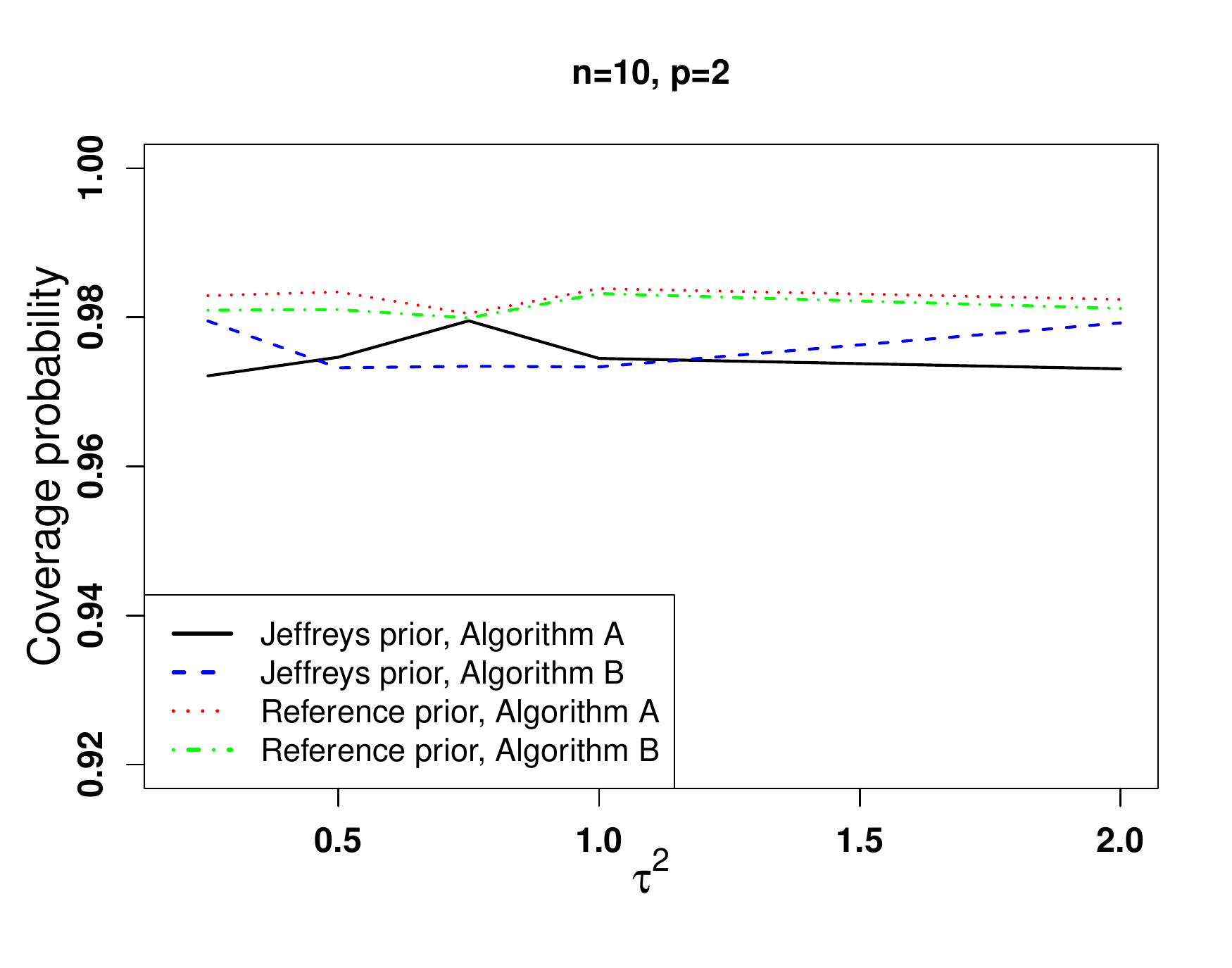}&\includegraphics[width=8cm]{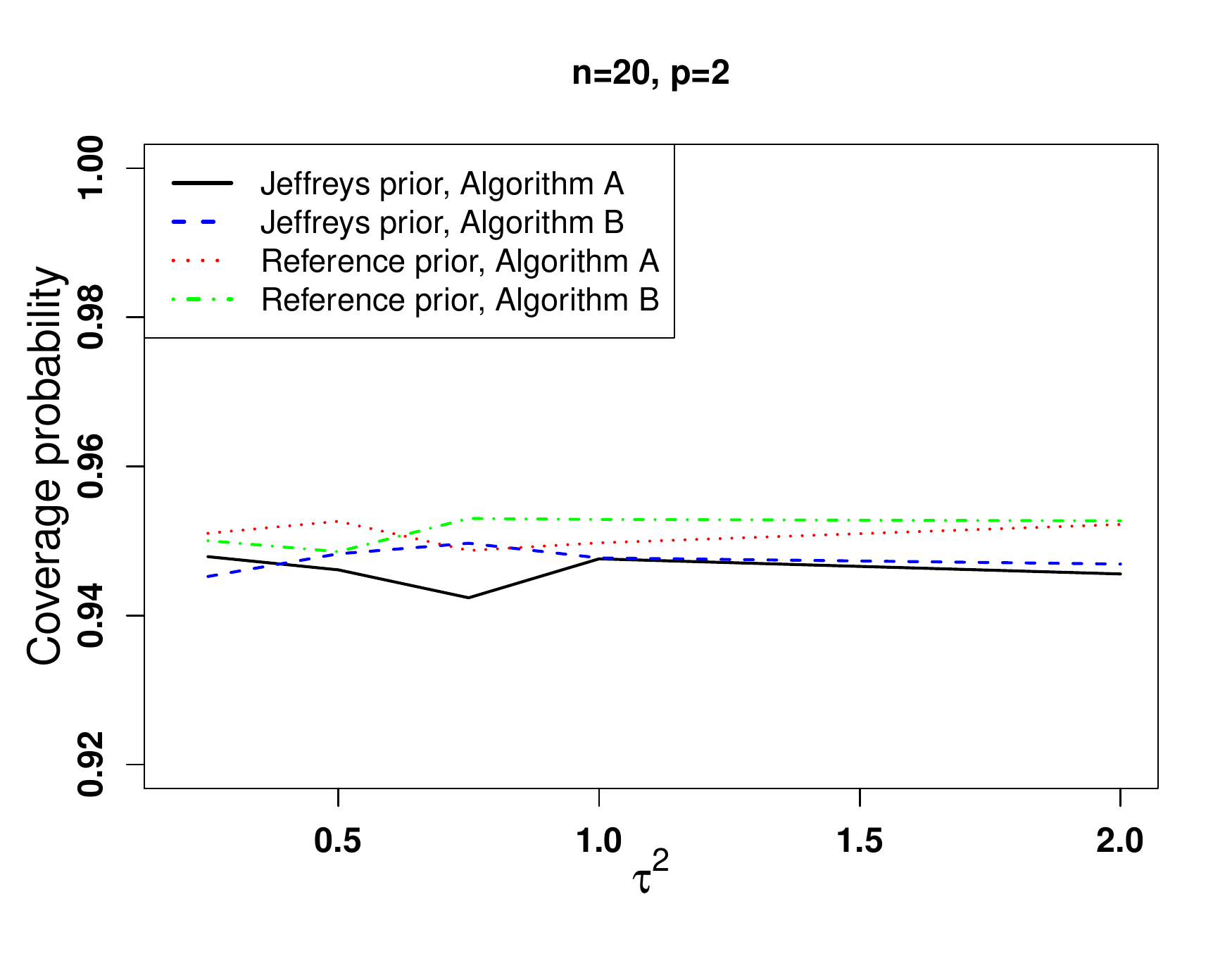}\\
\includegraphics[width=8cm]{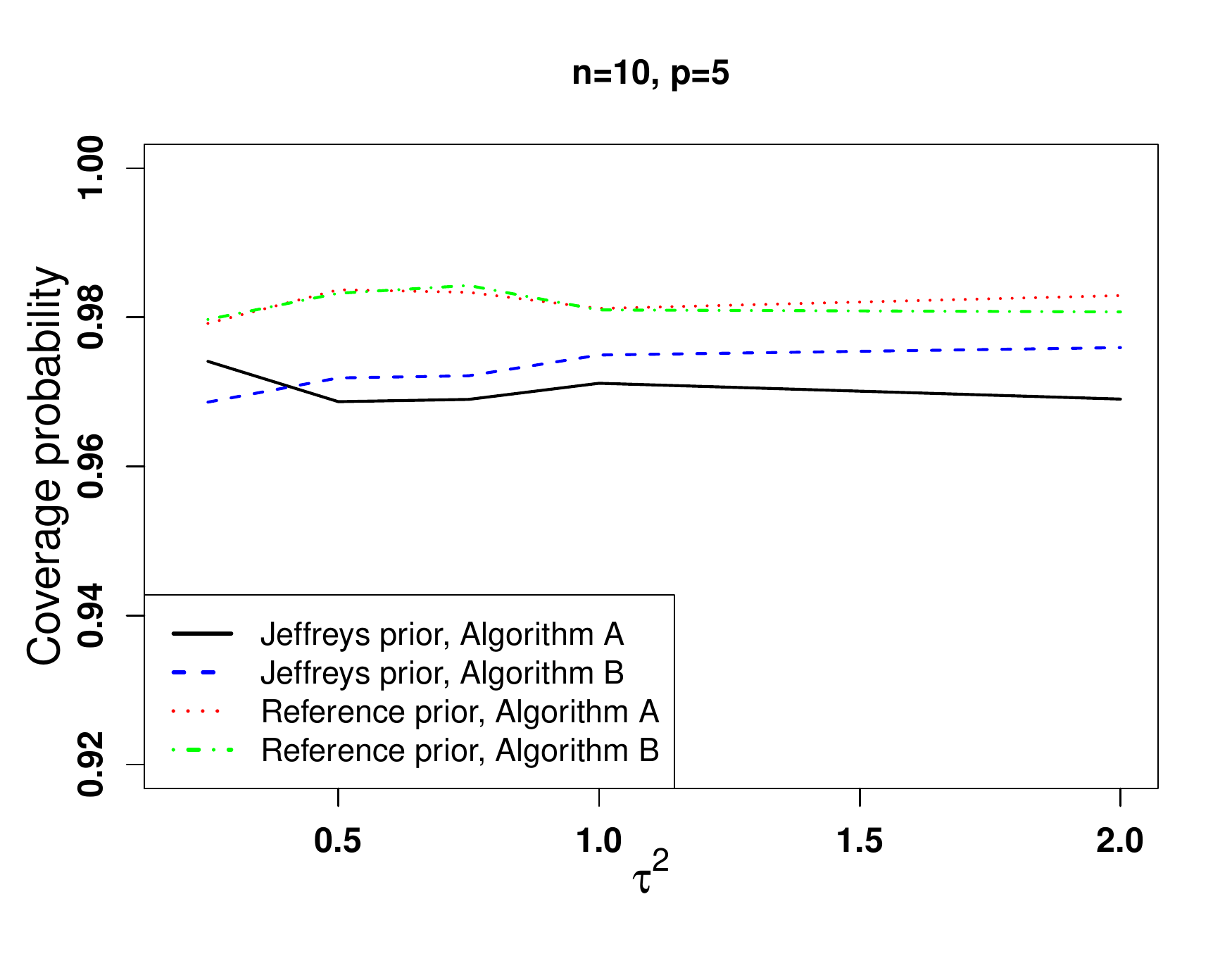}&\includegraphics[width=8cm]{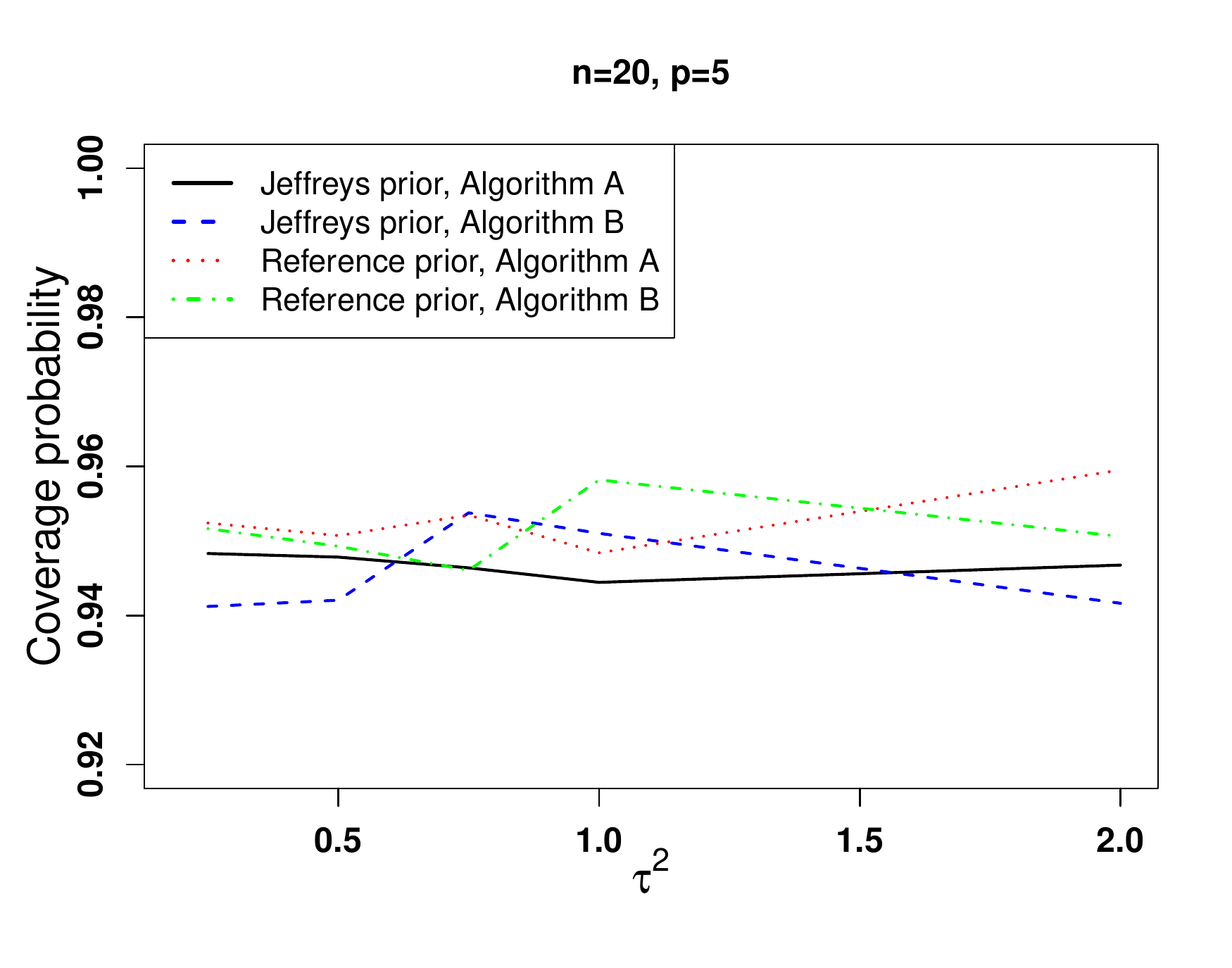}\\
\end{tabular}
 \caption{Coverage probabilities of the 95\% credible intervals for the first component of $\bmu$ as a function of $\tau^2$ under the assumption of the normal multivariate random effects model when the Berger and Bernardo reference prior and the Jeffreys prior are employed. We set $p\in\{2,5\}$ and $n\in\{10,20\}$. }
\label{fig:nor-mu1}
 \end{figure}


As a performance measure we use the empirical coverage probability of the credible interval constructed for $\mu_1$, which is computed based on 5000 independent repetitions. In each simulation run, the data matrix $\bX$ is drawn from the normal multivariate random effects model \eqref{mult-rem-nor} with the same $\bmu$, $\bPsi=\tau^2 \mathbf{\Xi}$, and $\bU=diag(\bU_1,...,\bU_p)$. The elements of $\bmu$ are generated from the uniform distribution on $[1,5]$. The eigenvalues of $\mathbf{\Xi}$, $\bU_1$, ... , $\bU_{p-1}$, and $\bU_p$ are generated from the uniform distribution on $[1,4]$, while the eigenvectors are simulated from the Haar distribution. The results in Figure \ref{fig:nor-mu1} are obtained for $p\in\{2,5\}$, $n\in\{10,20\}$, and $\tau^2 \in \{0.25, 0.5, 0.75, 1, 2\}$.

In Figure \ref{fig:nor-mu1} we observe that the credible intervals obtained by employing the Berger and Bernardo reference prior leads to wider credible intervals constructed for $\mu_1$, although the difference between the two non-informative priors is not large. The empirical coverage probabilities computed by using the two numerical procedures of drawing samples from the posterior distribution are above the chosen significance level of 95\% in almost all of the considered cases. Furthermore, when the sample size increases to $n=20$, then the constructed credible intervals possess almost perfect coverage probabilities of 95\% for both considered dimensions $p\in\{2,5\}$.

\subsection{$t$ multivariate random effects model}\label{sec:t}

In the case of the $t$ multivariate random effects model it holds that
\begin{equation}\label{f_t}
f(u)=K_{p,n,d}(1+u/d)^{-(pn+d)/2}~~ \text{with} ~~ K_{p,n,d}=(\pi d)^{-pn/2}\dfrac{\Gamma\left((d+pn)/2\right)}{\Gamma\left(d/2\right)}.
\end{equation}
Hence,
\begin{eqnarray*}
\pi(\bmu,\bPsi|\bX) &\propto& \frac{\pi(\bPsi)}{\sqrt{\prod_{i=1}^{n}\text{det}(\bPsi+\bU_i)}}\\
&\times& \Bigg(1+\frac{1}{d}(\bmu-\tilde{\bx}(\bPsi))^\top\left(\sum_{i=1}^{n}(\bPsi+ \bU_i)^{-1}\right)(\bmu-\tilde{\bx}(\bPsi))\\
&+& \frac{1}{d}\sum_{i=1}^{n} (\bx_i-\tilde{\bx}(\bPsi))^\top (\bPsi+ \bU_i)^{-1}(\bx_i-\tilde{\bx}(\bPsi))\Bigg)^{-(pn+d)/2}\\
&=&\frac{\pi(\bPsi)}{\sqrt{\prod_{i=1}^{n}\text{det}(\bPsi+\bU_i)}}
\left(1+\frac{1}{d}\sum_{i=1}^{n} (\bx_i-\tilde{\bx}(\bPsi))^\top (\bPsi+ \bU_i)^{-1}(\bx_i-\tilde{\bx}(\bPsi))\right)^{-(pn+d)/2}\\
&\times&
\Bigg(1+\frac{1}{pn+d-p}
\frac{pn+d-p}{d+\sum_{i=1}^{n} (\bx_i-\tilde{\bx}(\bPsi))^\top (\bPsi+ \bU_i)^{-1}(\bx_i-\tilde{\bx}(\bPsi))}\\
&\times&(\bmu-\tilde{\bx}(\bPsi))^\top\left(\sum_{i=1}^{n}(\bPsi+ \bU_i)^{-1}\right)(\bmu-\tilde{\bx}(\bPsi))\Bigg)^{-(pn+d)/2},
\end{eqnarray*}
which shows that the conditional posterior of $\bmu$ given $\bPsi$ is
\begin{eqnarray}\label{con_posterior_mu_t}
\pi(\bmu|\bPsi,\bX) &\propto& \Bigg(1+\frac{1}{pn+d-p}
\frac{pn+d-p}{d+\sum_{i=1}^{n} (\bx_i-\tilde{\bx}(\bPsi))^\top (\bPsi+ \bU_i)^{-1}(\bx_i-\tilde{\bx}(\bPsi))}\\
&\times&(\bmu-\tilde{\bx}(\bPsi))^\top\left(\sum_{i=1}^{n}(\bPsi+ \bU_i)^{-1}\right)(\bmu-\tilde{\bx}(\bPsi))\Bigg)^{-(pn+d)/2},
\end{eqnarray}
i.e., $\bmu$ conditionally on $\bPsi$ and $\bX$ has a $p$-dimensional $t$-distribution with $pn+d-p$ degrees of freedom, location parameter $\tilde{\bx}(\bPsi)$ and dispersion matrix
\[\frac{d+\sum_{i=1}^{n} (\bx_i-\tilde{\bx}(\bPsi))^\top (\bPsi+ \bU_i)^{-1}(\bx_i-\tilde{\bx}(\bPsi))}{pn+d-p}\left(\sum_{i=1}^{n}(\bPsi+ \bU_i)^{-1}\right)^{-1}.\]

Moreover, the marginal posterior for $\bPsi$ can also be deduced and it is expressed as
\begin{eqnarray}\label{marg_posterior_Psi_t}
\pi(\bPsi|\bX)&\propto& \frac{\pi(\bPsi)}{\sqrt{\text{det}(\sum_{i=1}^{n}(\bPsi + \bU_i)^{-1})}\prod_{i=1}^{n}\sqrt{\text{det}(\bPsi + \bU_i)}}\nonumber\\
&\times&\left(1+\frac{1}{d}\sum_{i=1}^{n} (\bx_i-\tilde{\bx}(\bPsi))^\top (\bPsi+ \bU_i)^{-1}(\bx_i-\tilde{\bx}(\bPsi))\right)^{-(pn+d)/2}\,.
\end{eqnarray}
To this end, the posterior mean vector and the covariance matrix of $\bmu$ are obtained as in \eqref{mu_post_mean} and \eqref{mu_post_var} with
\[C(\bPsi)=\frac{d+\sum_{i=1}^{n} (\bx_i-\tilde{\bx}(\bPsi))^\top (\bPsi+ \bU_i)^{-1}(\bx_i-\tilde{\bx}(\bPsi))}{pn+d-p-2}.\]

Finally, we note that the constant $J_2$ can analytically be computed in the case of the $t$ multivariate random effects model and it is expressed as (see, \citet[Section 3.2]{bodnar2019})
\[J_2=\frac{pn(pn+2)(pn+d)}{4(pn+2+d)}\]
The application of the last expression leads to the following formulas of the Berger and Bernardo reference prior
\begin{eqnarray}\label{prior_R_t}
\pi_R(\bPsi)&\propto& \Bigg(\text{det}\Bigg\{\mathbf{G}_p^\top\Bigg[ \frac{pn+d}{2(pn+2+d)}\sum_{i=1}^{n}\left((\bPsi + \bU_i)^{-1} \otimes (\bPsi + \bU_i)^{-1}\right)\\
&-&\frac{1}{2(pn+2+d)} \text{vec}\left(\sum_{i=1}^{n}(\bPsi + \bU_i)^{-1}\right)\text{vec}\left(\sum_{j=1}^{n}(\bPsi + \bU_j)^{-1}\right)^\top \Bigg]\mathbf{G}_p\Bigg\}
\Bigg)^{1/2},\nonumber
\end{eqnarray}
while the Jeffreys prior is given by \eqref{prior_J} with $\pi_R(\bPsi)$ as in \eqref{prior_R_t}. Furthermore, since $(1+u/d)^{-(pn+d)/2}$ is a decreasing function in $u$ and
\[\frac{J_2}{2pn+p^2n^2}-\frac{1}{4}=\frac{1}{4}\left(\frac{pn+d}{pn+d+2}-1\right)<0,\]
the posterior $\pi(\bmu,\bPsi|\bX)$ is proper for $n\ge p+1$ for the Berger and Bernardo reference prior and for $n \ge p$ for the Jeffreys prior due to Theorem~\ref{th4}.

Algorithm 1 and Algorithm 2 are used to draw samples from the posterior $\pi(\bmu,\bPsi|\bX)$ derived by employing the Berger and Bernardo reference prior and the Jeffreys prior, respectively. Under the special case of the $t$ multivariate random effects model, the step (ii) of the both algorithm is performed by computing $\bPsi^{(b)}=\frac{\xi^{(b)}}{d} \mathbf{\Omega}^{(b)}$ where $\xi^{(b)}$ and $\mathbf{\Omega}^{(b)}$ are simulated independently from $\chi^2_d$-distribution and the inverse Wishart distribution with parameter matrix $(n-1)\sum_{i=1}^n(\bx_i - \bmu^{(b)})(\bx_i - \bmu^{(b)})^\top$ and degrees of freedom equal to $n+p+1$ under the Berger and Bernardo reference prior and $n+p+2$ under the Jeffreys prior.

The modification of Algorithms 1 and 2 similar to the one derived for the normal multivariate random effects model can also be obtained under the assumption of the $t$-distribution. The application of the equality
\begin{eqnarray*}
&&\left(1+\frac{1}{d}\text{tr}\left(\bPsi^{-1}\sum_{i=1}^n(\bx_i - \bmu)(\bx_i - \bmu)^\top \right)\right)^{-(pn+d)/2}
=\left(1+\frac{n-1}{d}\text{tr}\left(\bPsi^{-1}\bS\right)\right)^{-(pn+d)/2}\\
&\times&\left(1+\frac{1}{np+d-p}\frac{n(np+d-p)}{d+(n-1)\text{tr}\left(\bPsi^{-1}\bS\right)} (\bmu-\bar{\bx})^\top\bPsi^{-1}(\bmu-\bar{\bx})\right)^{-(pn+d)/2}
\end{eqnarray*}
leads to another numerical procedure described in Algorithm 4 when the posterior $\pi(\bmu,\bPsi|\bX)$ is derived by applying the the Berger and Bernardo reference prior. Under the Jeffreys prior the step (2) of Algorithm 4 should be modified by generating $\mathbf{\Omega}^{(w)}$ from $IW_{p}(n+p+1,(n-1)\bS)$.

\begin{algorithm}
\caption{Metropolis-Hastings algorithm for drawing realizations from $\pi(\bmu, \bPsi|\bX)$ as in \eqref{post_joint} under the Berger and Bernardo reference prior \eqref{prior_R} in the $t$ multivariate random effects model}\label{algorithm_R_t}
\begin{enumerate}[(1)]
\item \textbf{Initialization:} Choose the initial values $\bmu^{(0)}$ and $\bPsi^{(0)}$ for $\bmu$ and $\bPsi$ and set $b =0$.
\item \textbf{Generating new values of $\bmu^{(w)}$ and $\bPsi^{(w)}$ from the proposal:}
\begin{enumerate}[(i)]
\item For given data $\bX=(\bx_1,...,\bx_n)$, generate $\bPsi^{(w)}=\frac{\xi^{(w)}}{d} \mathbf{\Omega}^{(w)}$ where $\xi^{(w)}$ and $\mathbf{\Omega}^{(w)}$ are simulated independently with from $\xi^{w}\sim\chi^2_d$ from $\mathbf{\Omega}^{(w)}|\bX \sim IW_{p}(n+p,(n-1)\bS)$;
\item Using data $\bX$ and the drawn in step (i) $\bPsi^{(w)}$, generate $\bmu^{(w)}$ from $\bmu|\bPsi=\bPsi^{(w)}, \bX \sim t_p\left(pn+d-p,\bar{\bx},\frac{d+(n-1)\text{tr}\left((\bPsi^{(w)})^{-1}\bS\right)}{n(np+d-p)}\bPsi^{(w)}\right)$ with $\bar{\bx}$ and $\bS$ as in \eqref{barx_bS} and the symbol $t_p(m,\mathbf{r},\mathbf{R})$ stands for the multivariate $p$-dimensional $t$-distribution with $m$ degrees of freedom, location vector $\mathbf{r}$, and scale matrix $\mathbf{R}$.
\end{enumerate}
\item \textbf{Computation of the Metropolis-Hastings ratio:}
\[MH^{(b)}=\frac{\pi(\bmu^{(w)}, \bPsi^{(w)}|\bX)q_R(\bmu^{(b-1)},\bPsi^{(b-1)}| \bX)}{\pi(\bmu^{(b-1)}, \bPsi^{(b-1)}|\bX)q_R(\bmu^{(w)},\bPsi^{(w)}| \bX)}.\]
\item \textbf{Moving to the next state of the Markov chain:}
\begin{enumerate}[(i)]
\item Generate $U^{(b)}$ from the uniform distribution on $[0,1]$;
\item If $U^{b}<\min\left\{1,MH^{(b)}\right\}\pi(\bmu^{(b)}$, then set $\bmu^{(b)}=\bmu^{(w)}$ and $\bPsi^{(b)}=\bPsi^{(w)}$ (Markov chain moves to the new state). Otherwise, set $\bmu^{(b)}=\bmu^{(b-1)}$ and $\bPsi^{(b)}=\bPsi^{(b-1)}$ (Markov chain stays in the previous state).
\end{enumerate}
\item Return to step (2), increase $b$ by 1, and repeat until the sample of size $B$ is accumulated.
\end{enumerate}
\end{algorithm}

In order to investigate the properties of the two proposed algorithms, we conduct a simulation study for the t multivariate random effects model designed similarly to the one presented in Section \ref{sec:normal} for the normal multivariate random effect models. Also, we use the same notations 'Algorithm A' and 'Algorithm B' to distinguish between the two procedures to draw the sample from the posterior distributions derived by employing the Berger and Bernardo reference prior and the Jeffreys prior.

 \begin{figure}[h!t]
 \centering
\begin{tabular}{cc}
\includegraphics[width=8cm]{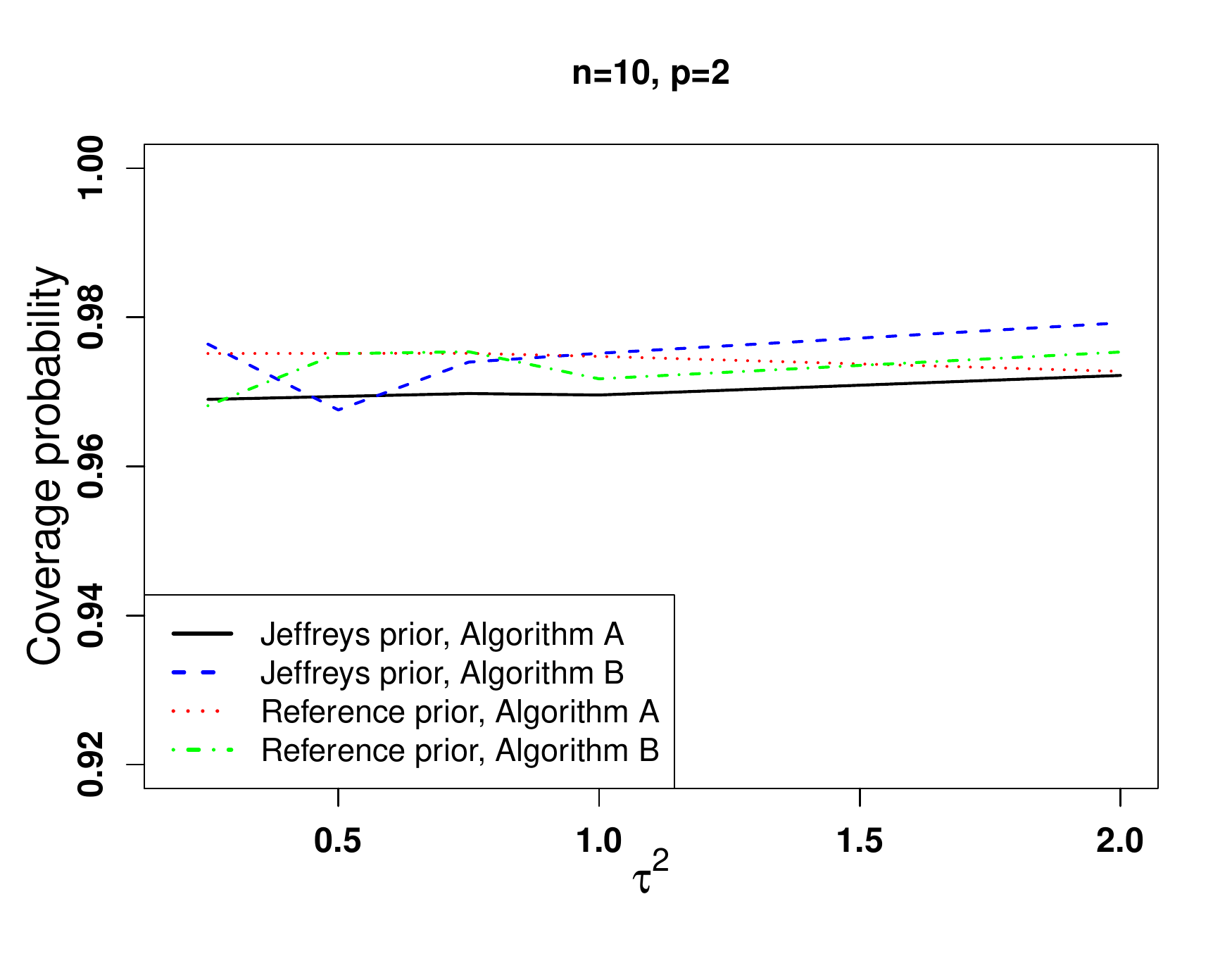}&\includegraphics[width=8cm]{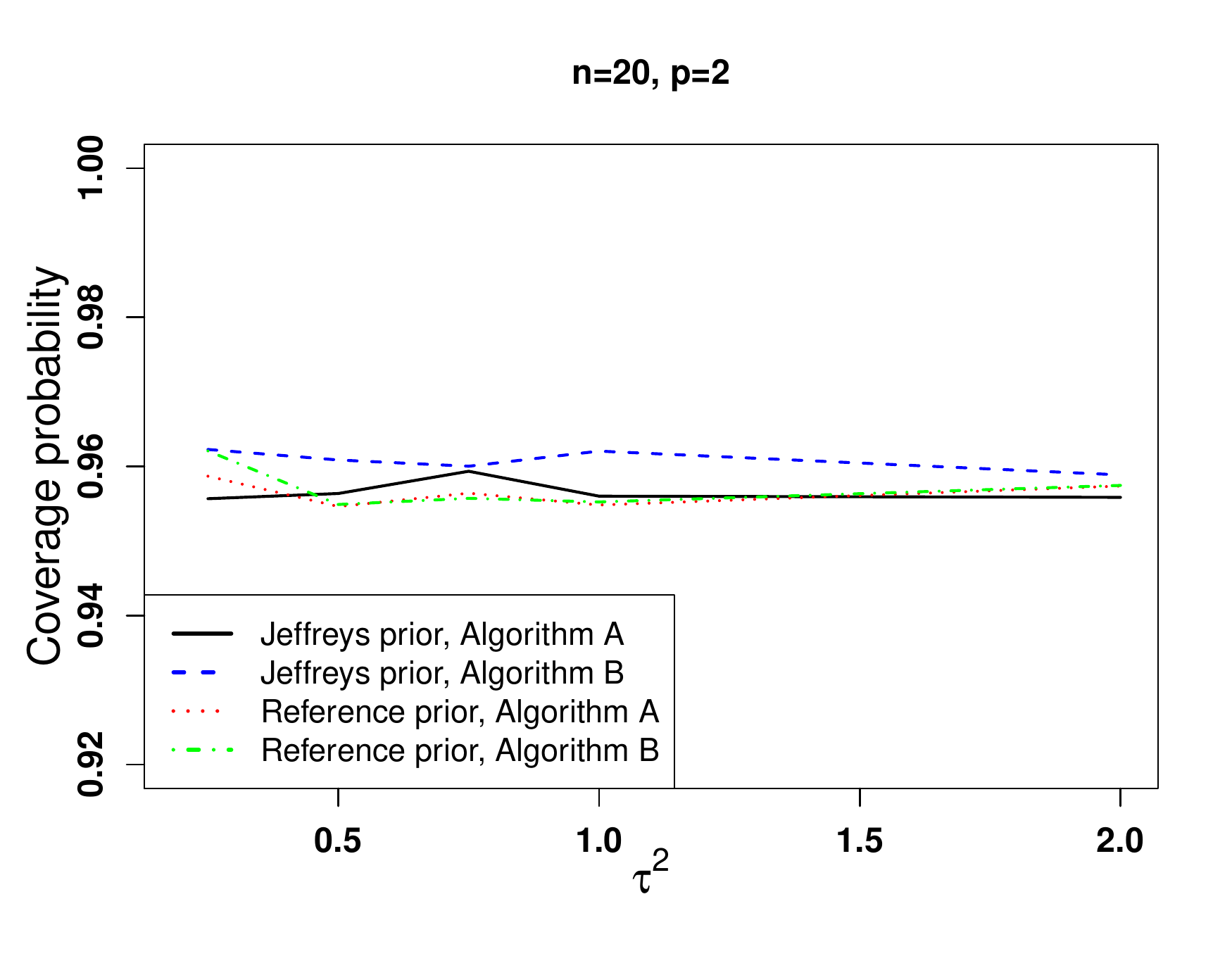}\\
\includegraphics[width=8cm]{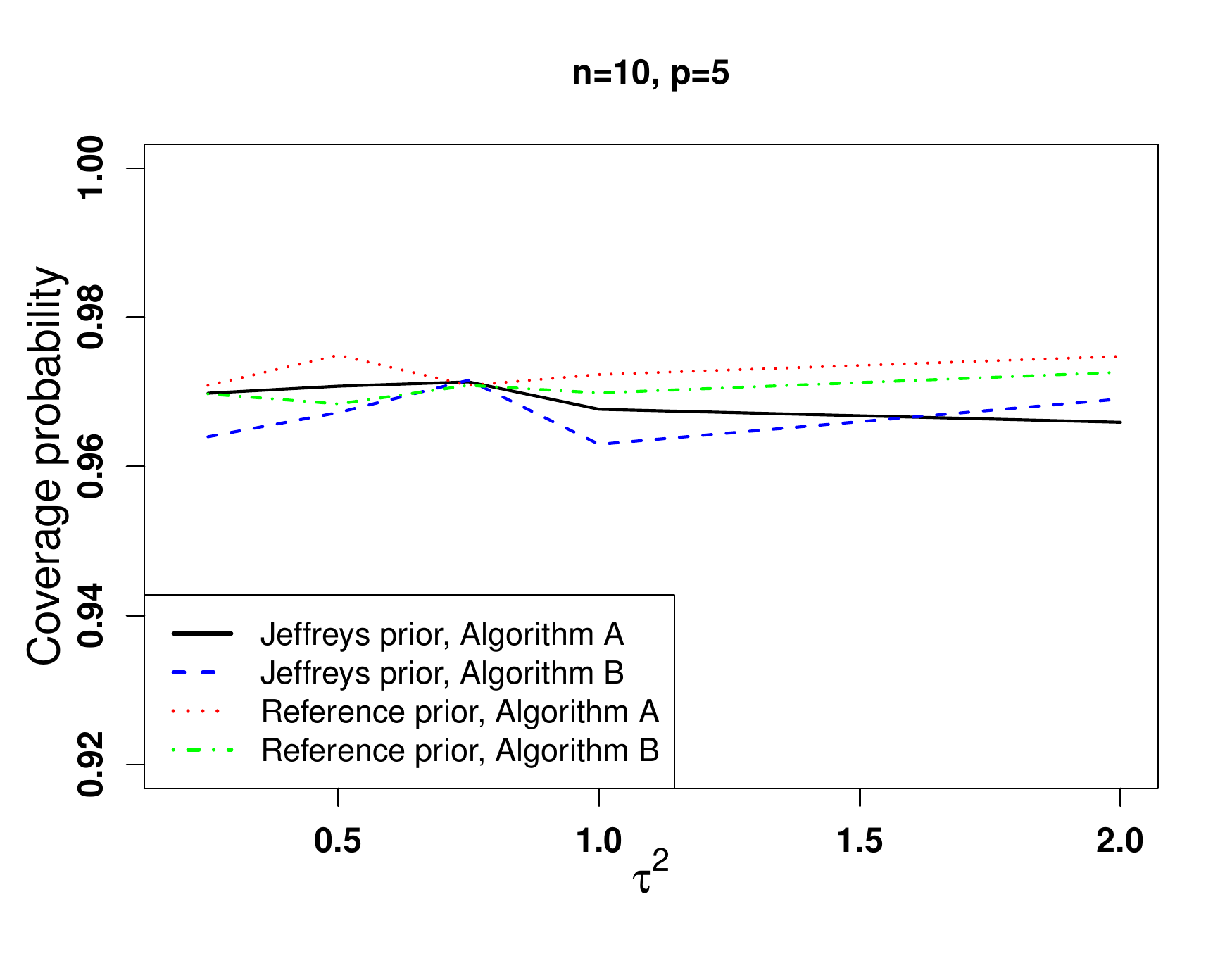}&\includegraphics[width=8cm]{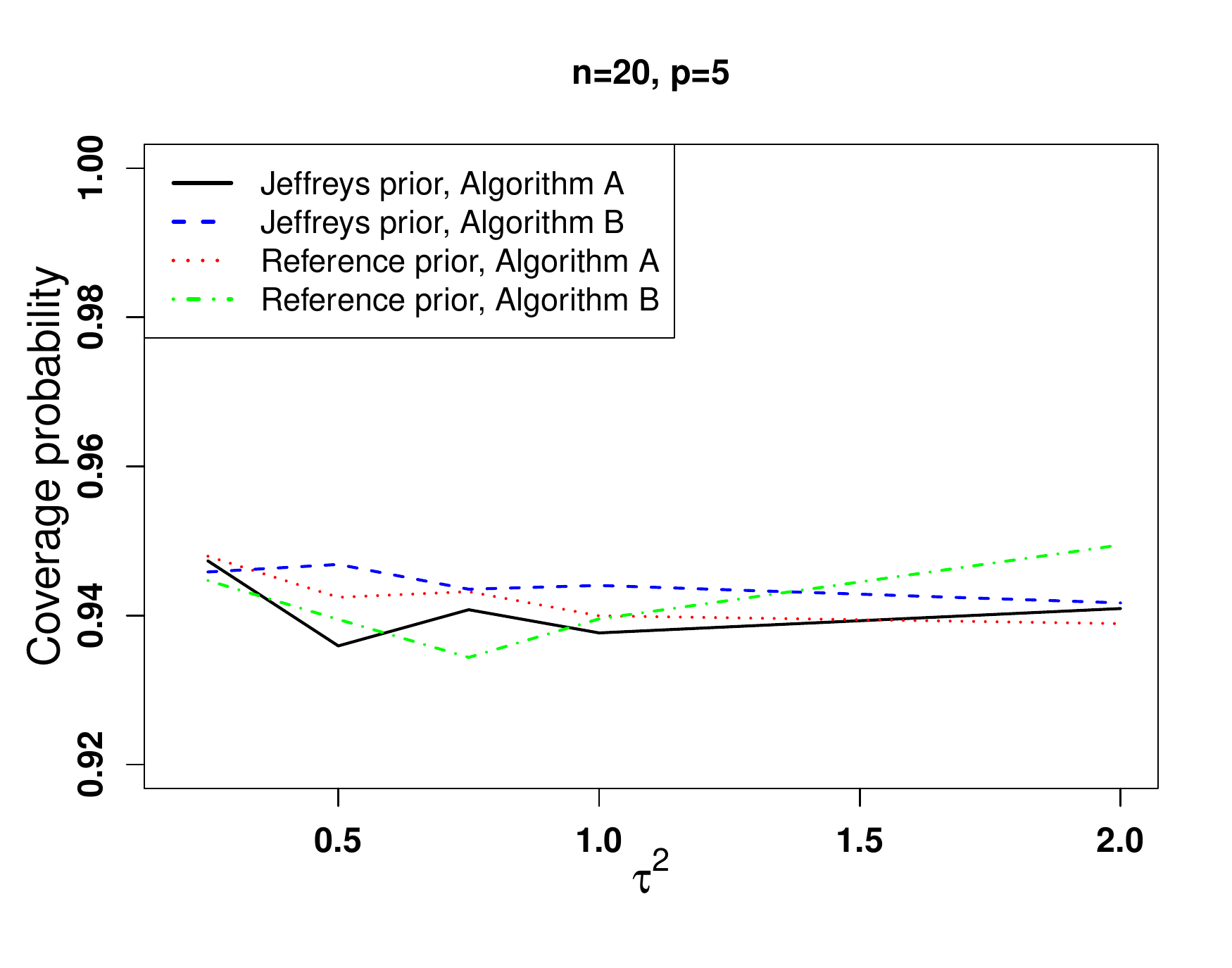}\\
\end{tabular}
 \caption{Coverage probabilities of the 95\% credible intervals for the first element of $\bmu$ as a function of $\tau^2$ under the assumption of the $t$ multivariate random effects model when the Berger and Bernardo reference prior and the Jeffreys prior are employed. We set $p \in \{2,5\}$ and $n\in \{10,20\}$.}
\label{fig:t-mu1}
 \end{figure}


We use the empirical coverage probability of the credible interval constructed for $\mu_1$ as a performance measure and compute it based on 5000 independent repetitions. In each simulation run, the data matrix $\bX$ is simulated from the t multivariate random effects model, i.e., from the model \eqref{mult-rem} with $f(.)$ as in \eqref{f_t}. The model parameters $\bmu$, $\bPsi=\tau^2 \mathbf{\Xi}$, and $\bU=diag(\bU_1,...,\bU_p)$ are chosen in the same way as the corresponding parameters of the normal multivariate random effects model in Section \ref{sec:normal}. Finally, we set $p\in\{2,5\}$, $n\in\{10,20\}$, and $\tau^2 \in \{0.25, 0.5, 0.75, 1, 2\}$.

The results of the simulation study are depicted in Figure \ref{fig:t-mu1}. Similarly to the findings obtained for the normal multivariate random effects model, the empirical coverage probabilities are larger than the chosen significance level of 95\% in almost all of the considered cases, independently whether the Markov chains are constructed following Algorithm A or Algorithm B. The application of the Berger and Bernardo reference prior leads to a slightly larger values of the empirical coverage probabilities. Finally, the coverage probabilities are close to 95\% when the sample size is $n=20$ independently of the chosen values of $p \in \{2,5\}$.

\section{Empirical illustration}\label{sec:emp}

In this section we illustrate the derived theoretical findings on real data consisting of results obtained in ten studies that assess the effectiveness of hypertension treatment for reducing blood pressure. The treatment effects on the systolic blood pressure and diastolic blood pressure are investigated in the studies where the negative values document positive effect of the treatment. The data are provided in Table \ref{tab:data} and are taken from \citet{jackson2013matrix} where the treatment effects in each study are provided together with the covariance matrices $\bU_i$ which are assumed to be known throughout this section.

\begin{table}[h!t]
\centering
\begin{tabular}{c|c|c|c|c|c}
  \hline \hline
 Study & $X_{i;1}$ (SBP) & $X_{i;2}$ (DBP) & $\sqrt{U_{i;11}}$ (SBP) & $\rho_{i;12}=\dfrac{U_{i;12}}{\sqrt{U_{i;11}U_{i;22}}}$ & $\sqrt{U_{i;22}}$ (DBP)\\
\hline
 1&  -6.66& -2.99& 0.72 & 0.78& 0.27 \\
 2& -14.17& -7.87& 4.73 & 0.45& 1.44 \\
 3& -12.88& -6.01&10.31 & 0.59& 1.77 \\
 4& -8.71 & -5.11& 0.30 & 0.77& 0.10 \\
 5& -8.70 & -4.64& 0.14 & 0.66& 0.05 \\
 6& -10.60& -5.56& 0.58 & 0.49& 0.18 \\
 7& -11.36& -3.98& 0.30 & 0.50& 0.27 \\
 8& -17.93& -6.54& 5.82 & 0.61& 1.31 \\
 9& -6.55 & -2.08& 0.41 & 0.45& 0.11 \\
10& -10.26& -3.49& 0.20 & 0.51& 0.04 \\
\hline\hline
\end{tabular}
\caption{Data collected in 10 studies about the effectiveness of hypertension treatment with the aim to reduce blood pressure. The variables $X_{i;1}$ and $X_{i;2}$ denote the treatment effects on the systolic blood pressure (SBP) and the diastolic blood pressures (DBP) from the $i$th study, while $\mathbf{U}_i=(U_{i;lj})_{lj=1,2}$ is the corresponding covariance matrix.}
\label{tab:data}
\end{table}

Multivariate meta-analysis is performed by using data from Table \ref{tab:data} under the assumption of the normal multivariate random effects model (Section \ref{sec:normal}) and the $t$ multivariate random effects model (Section \ref{sec:t}) when the Berger and Bernardo reference prior and the Jeffreys prior are employed. The samples from the joint posterior distribution $\pi(\bmu,\bPsi|\bX)$ are drawn by applying two versions of the Metropolis-Hastings algorithm which are described in Section \ref{sec:normal} for the normal multivariate random effects model and denoted by Algorithm A and Algorithm B, respectively. For each type of the Metropolis-Hastings algorithm, the distributional class of the multivariate random effects model, and the chosen prior, $10^5$ realizations from the posterior distribution $\pi(\bmu,\bPsi|\bX)$ are drawn with 10\% used as burn-in sample.

\begin{table}[h!t]
\centering
\begin{tabular}{c||c|c||c|c}
  \hline \hline
  & \multicolumn{2}{c||}{Normal random effects model}&\multicolumn{2}{c}{t random effects model}\\
  \cline{2-5}
& $\mu_{1}$ (SBP) & $\mu_{2}$ (DBP) & $\mu_{1}$ (SBP) & $\mu_{2}$ (DBP)\\
  \hline
  \multicolumn{5}{c}{Jeffreys prior, Algorithm A}\\
  \hline
  post. mean & -9.79 & -4.05 & -10.15 & -4.67  \\
  post. median & -9.60 & -4.27 & -10.10 & -4.66 \\
  post. sd. & 0.88 & 0.93 & 1.08 & 0.60  \\
  cred. inter. & [-11.73, -8.00] & [-5.61, -2.66] & [-12.44,-8.11] & [-5.87,-3.51]\\
  \hline
  \multicolumn{5}{c}{Jeffreys prior, Algorithm B}\\
  \hline
  post. mean & -9.78 & -4.37 & -10.03 & -4.66  \\
  post. median & -9.84 & -4.37 & -9.97 & -4.65  \\
  post. sd. & 0.74 & 0.50 & 1.13 & 0.61  \\
  cred. inter. & [-11.46, -8.39] & [-5.38, -3.38] & [-12.46, -7.99] & [-5.90, -3.51]  \\
  \hline
  \multicolumn{5}{c}{Berger and Bernardo reference prior, Algorithm A}\\
  \hline
  post. mean & -9.81 & -4.49 & -10.11 & -4.67  \\
  post. median & -9.87 & -4.44 & -10.06 & -4.66  \\
  post. sd. & 1.04 & 0.59 & 1.16 & 0.64  \\
  cred. inter. & [-12.06, -8.00] & [-5.78, -3.42] & [-12.58, -7.97] & [-5.96, -3.43] \\
  \hline
  \multicolumn{5}{c}{Berger and Bernardo reference prior, Algorithm B}\\
  \hline
  post. mean & -9.70 & -4.51 & -10.08 & -4.68 \\
  post. median & -9.72 & -4.53 & -10.03 & -4.65 \\
  post. sd. & 1.01 & 0.58 & 1.13 & 0.64 \\
  cred. inter. & [-11.88, -8.06] & [-5.67, -3.49] &[-12.50, -7.97] & [-5.94, -3.42] \\
  \hline
  \multicolumn{5}{c}{REML}\\
  \hline
  estimator & -9.50 & -4.43 & -- & -- \\
  stand. error & 0.77 & 0.48 & -- & -- \\
  \hline
  \multicolumn{5}{c}{Method of moments, \citet{jackson2010extending} }\\
  \hline
  estimator & -9.13 & -4.30 & -- & --  \\
  stand. error & 0.54 & 0.36 & -- & -- \\
  \hline
  \multicolumn{5}{c}{Method of moments, \citet{jackson2013matrix}}\\
  \hline
  estimator & -9.17 & -4.31 & -- & --  \\
  stand. error & 0.55 & 0.36 & -- & -- \\
\hline \hline
\end{tabular}
\caption{Results of Bayesian inference (posterior mean, posterior median, posterior standard deviation, 95\% credible interval) for the parameters of the normal multivariate random effects model obtained for the data from Table \ref{tab:data} by employing the Berger and Bernardo reference prior and the Jeffrey prior. The samples from the posterior distributions are drawn by Algorithm A and Algorithm B defined in Section \ref{sec:normal}. The last three panels of the table include the results of the restrictive maximum likelihood estimator of \citet{chen2012method}, and two methods of moment estimators from \citet{jackson2010extending} and \citet{jackson2013matrix}. }
\label{tab:normal-t}
\end{table}

The first two columns of Table \ref{tab:normal-t} present the results for the normal multivariate random effects model, while the results for the $t$ multivariate random effects model with $d=3$ degrees of freedom are shown in the third and the fourth columns of the table. For each chosen prior, random effects model, and numerical algorithm to draw a sample from the posterior distribution, we compute the posterior mean and posterior median as two Bayesian point estimators for the overall mean vector together with the posterior standard deviation and 95\% probability symmetric credible interval. In the case of the $t$ multivariate random effects model we multiply $\bU_i$ by $\frac{d-2}{d}$ in order to ensure that the within-study covariance matrix calculated under the assumption of the $t$ multivariate random effects model coincides with the one given in Table \ref{tab:data}. Finally, for comparison purposes, we also include the results obtained by three approaches of the frequentist statistics which are given in Table 3 of \citet{jackson2013matrix}.

 \begin{figure}[h!t]
\centering
\begin{tabular}{cc}
\includegraphics[width=8cm]{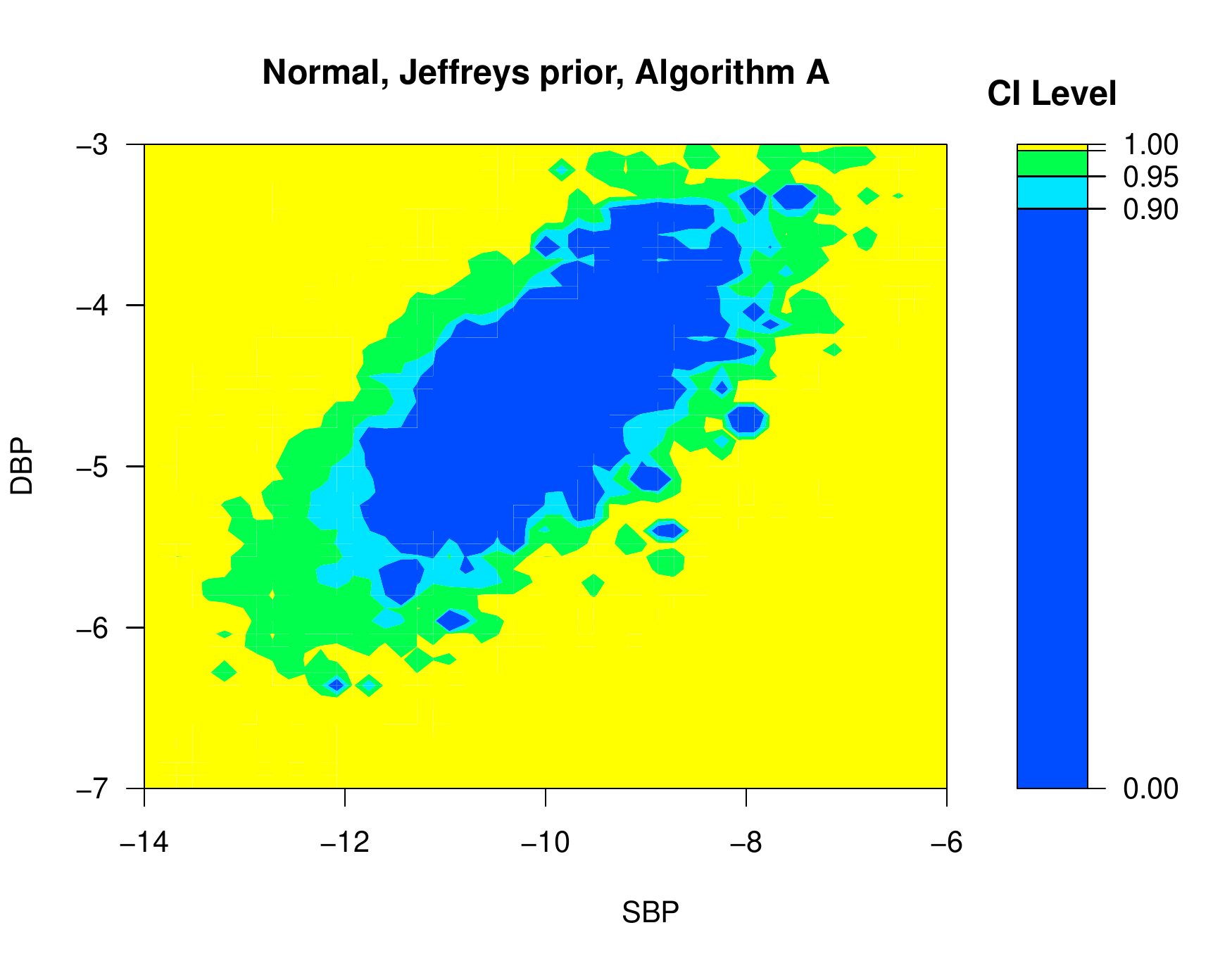}&\includegraphics[width=8cm]{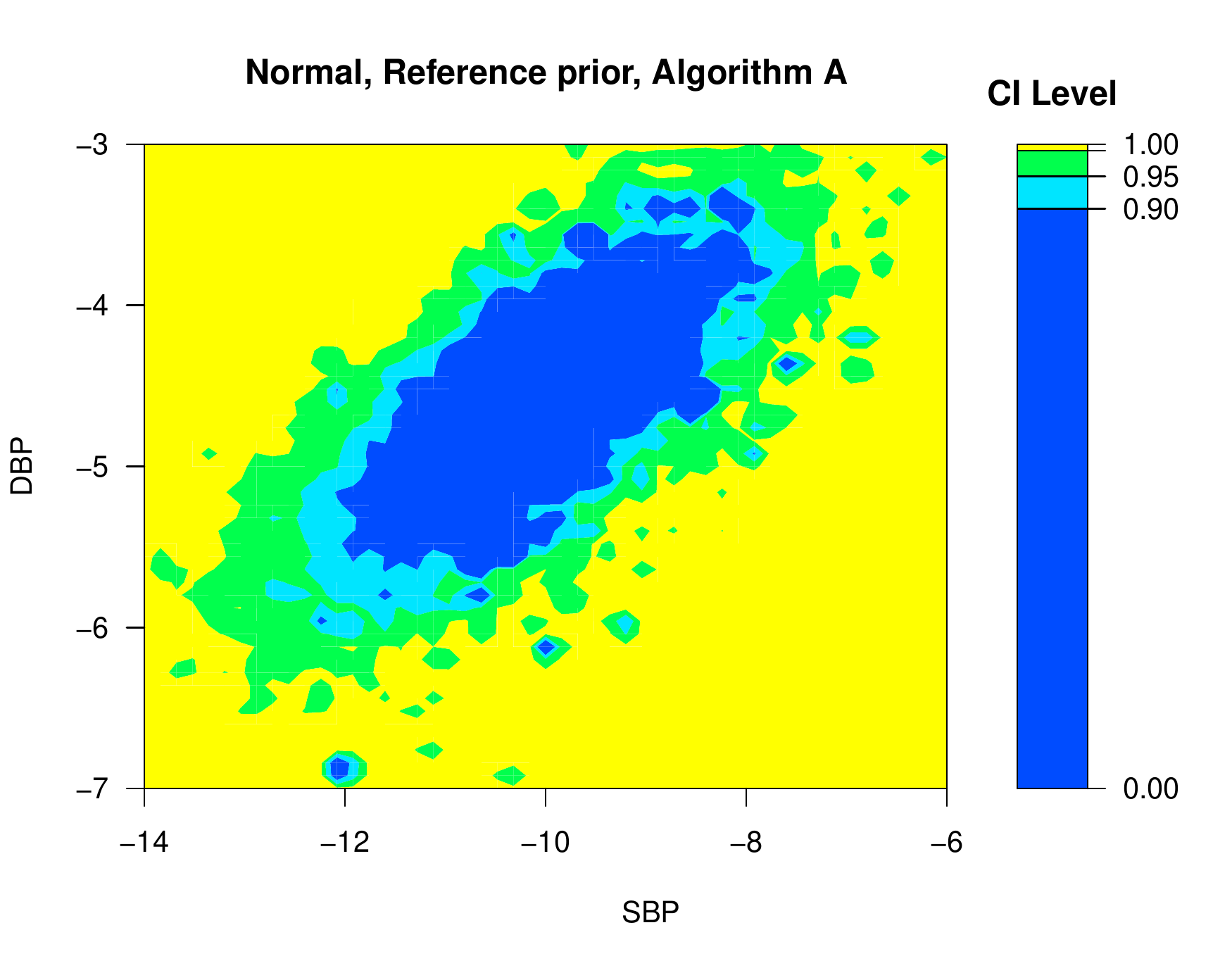}\\
\includegraphics[width=8cm]{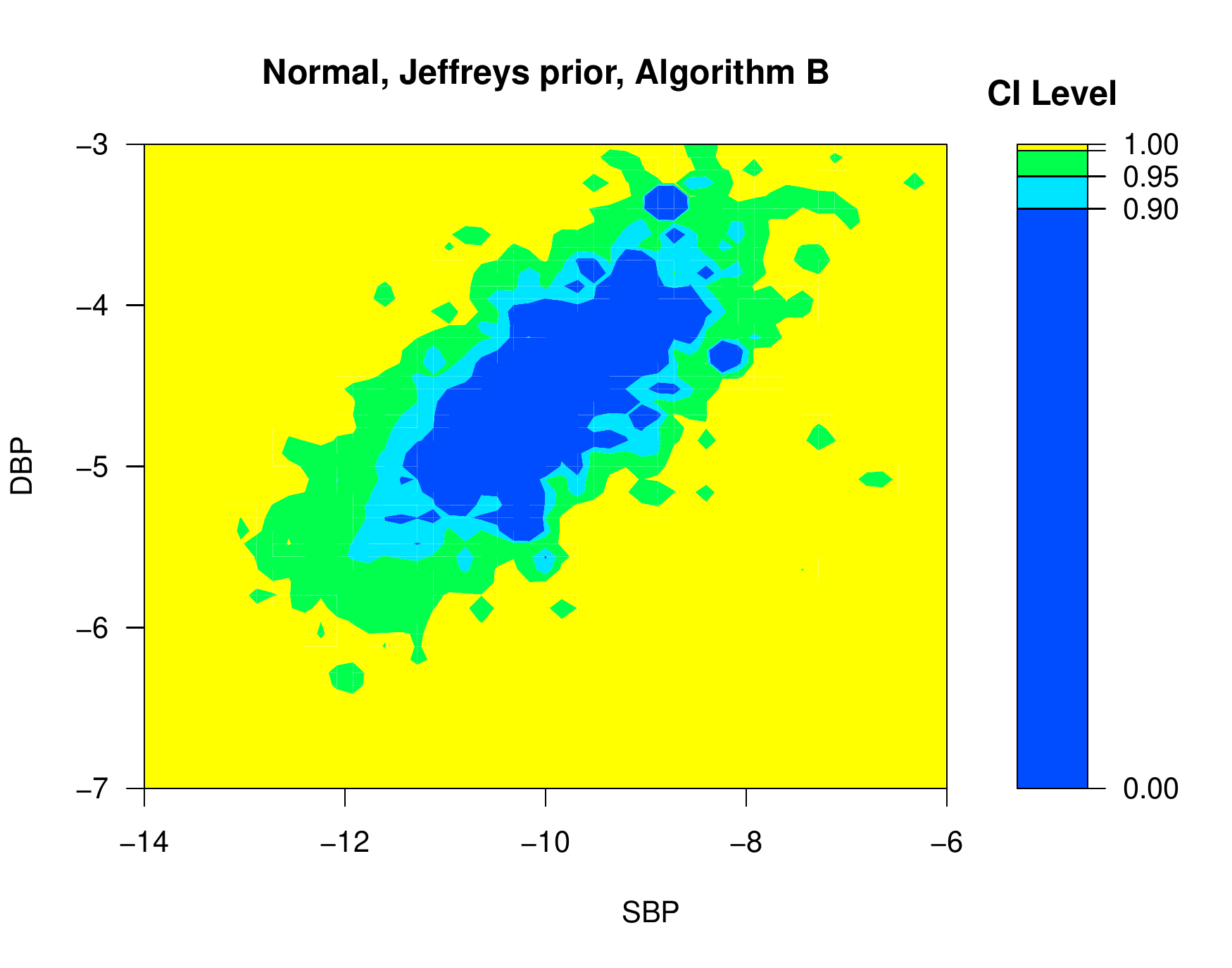}&\includegraphics[width=8cm]{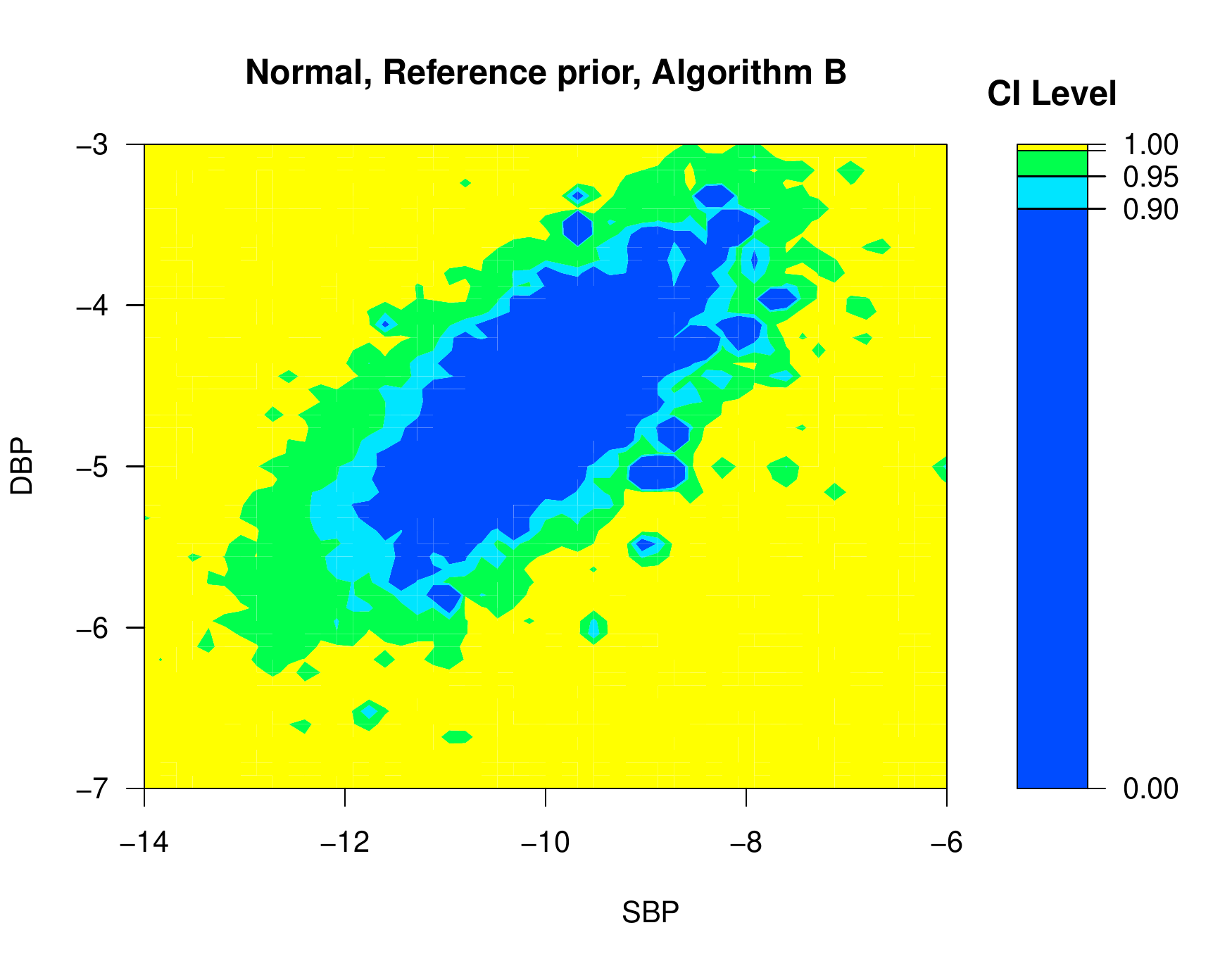}\\
\end{tabular}
 \caption{Credible sets for $\mu_1$ (SBP) and $\mu_2$ (DBP) obtained from the posterior distribution $\pi(\bmu|\bX)$ derived for the location parameters of the normal multivariate random effects model by employing the Berger and Bernardo reference prior and the Jeffrey prior and using data from Table \ref{tab:data}. The samples from the posterior distributions are drawn by Algorithm A and Algorithm B defined in Section \ref{sec:normal}.}
\label{fig:norm-emp-study}
 \end{figure}

All Bayesian point estimators derived under the assumption of the normal multivariate random effects model are very similar and they are almost always slightly smaller than those obtained by the frequentist approaches. In contrast, the computed Bayesian standard errors are larger than those computed by the frequentist approaches, especially when the two methods of moments are used in their computation. These results are in line with statistical theory and reflect the fact that the Bayesian methods in contrast to the frequentist approaches take automatically the uncertainty about the between-study covariance matrix $\bPsi$ into account, while the frequentist methods usually ignore that $\bPsi$ is an unknown nuisance parameter of the model which has to be estimated before the inferences for the overall mean vector are constructed. Finally, we have that the credible intervals obtained by employing the Berger and Bernardo reference prior are wider than those obtained by using the Jeffreys prior. Similar results are also obtained in the simulation study of Section \ref{sec:normal} (see, Figure \ref{fig:nor-mu1}), where the larger values of the coverage probabilities are documented for the Berger and Bernardo reference prior. Such a result was also documented in the univariate case in \citet{bodnar2019}. Finally, in the case of the $t$ multivariate random effects model, the estimated elements of the overall mean vector $\bmu$ become even smaller than those observed under the of the normal multivariate random effects model, while the corresponding Bayesian standard deviations increase reflecting the impact of heavy tails of the $t$-distribution.

\begin{figure}[h!t]
\centering
\begin{tabular}{cc}
\includegraphics[width=8cm]{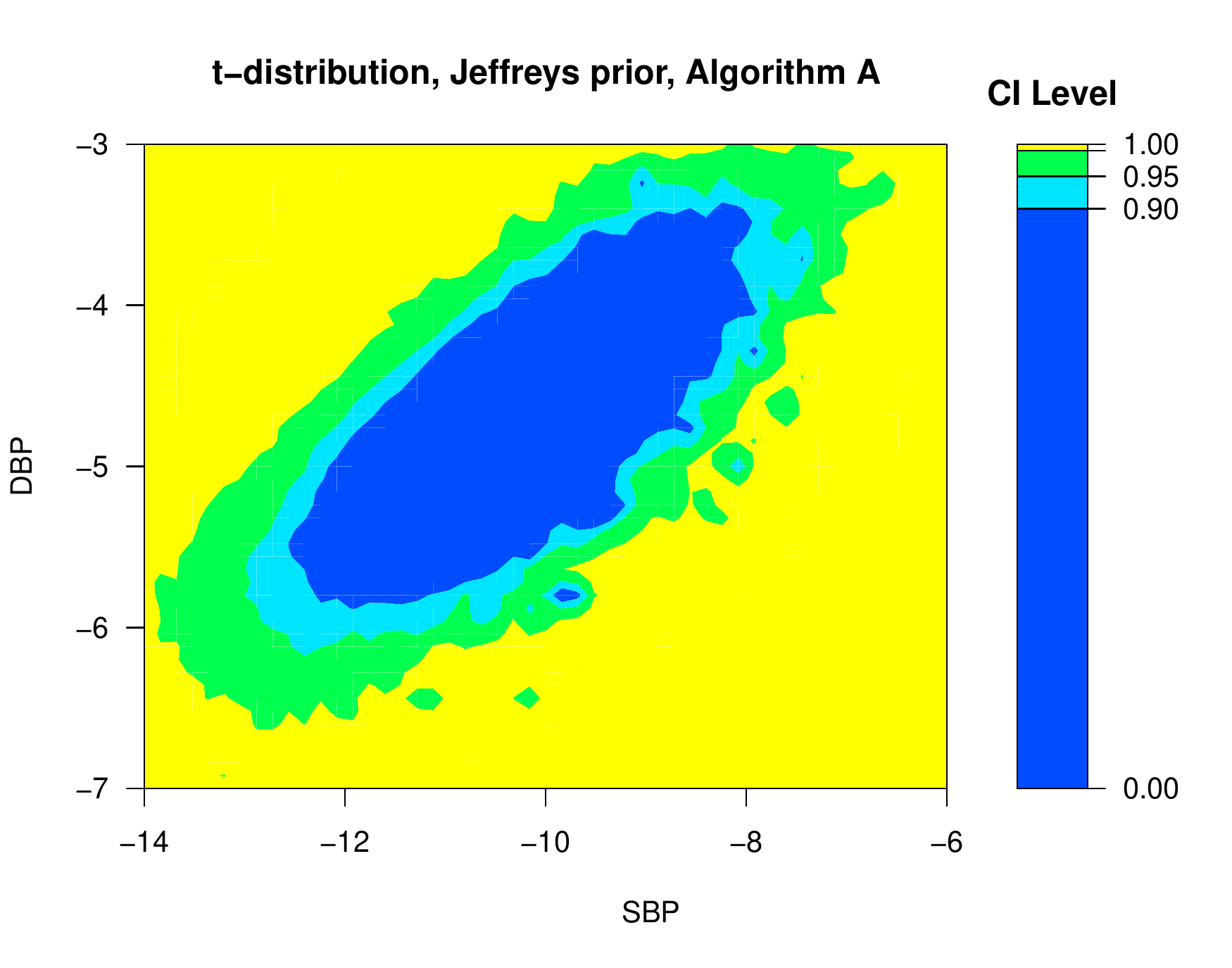}&\includegraphics[width=8cm]{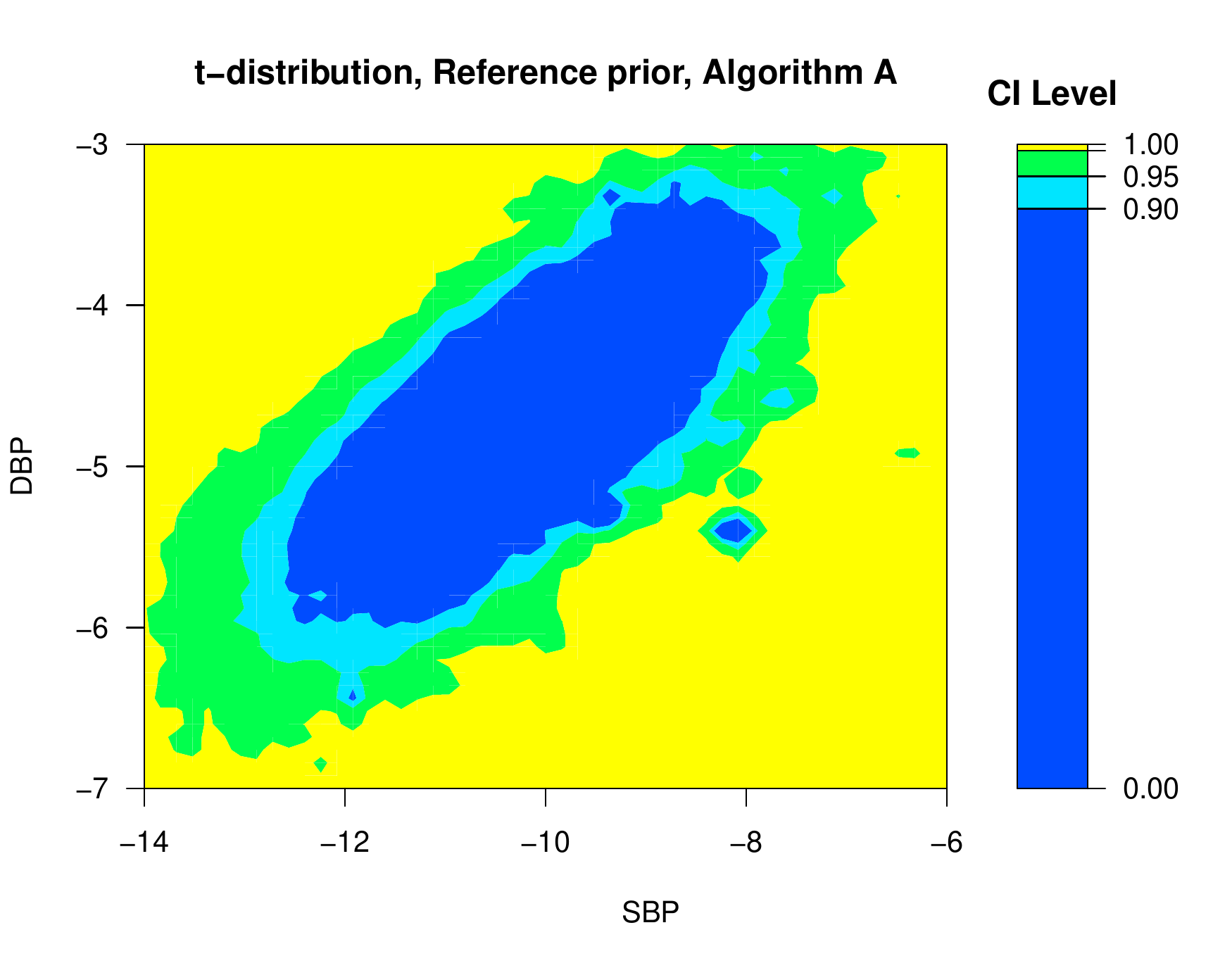}\\
\includegraphics[width=8cm]{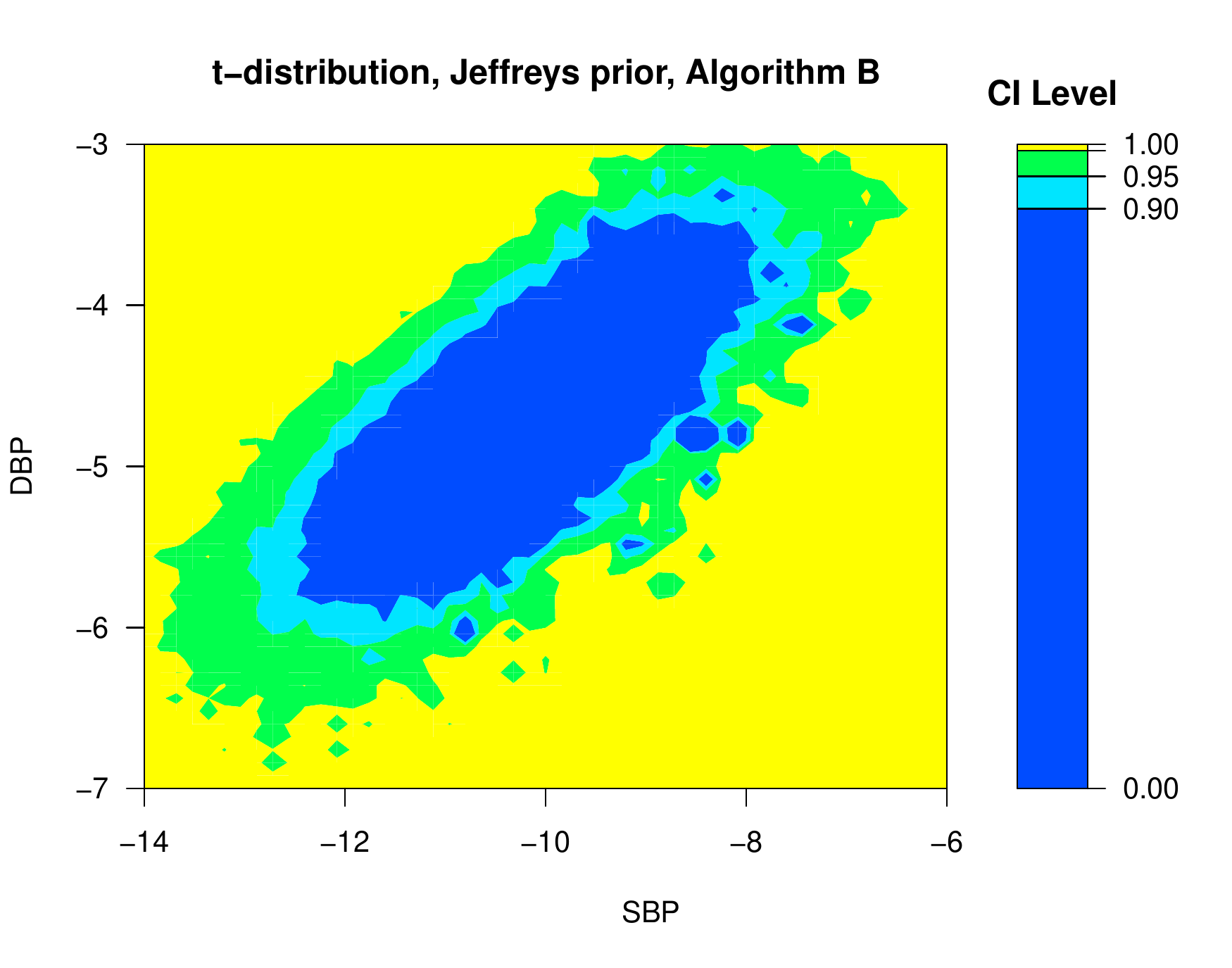}&\includegraphics[width=8cm]{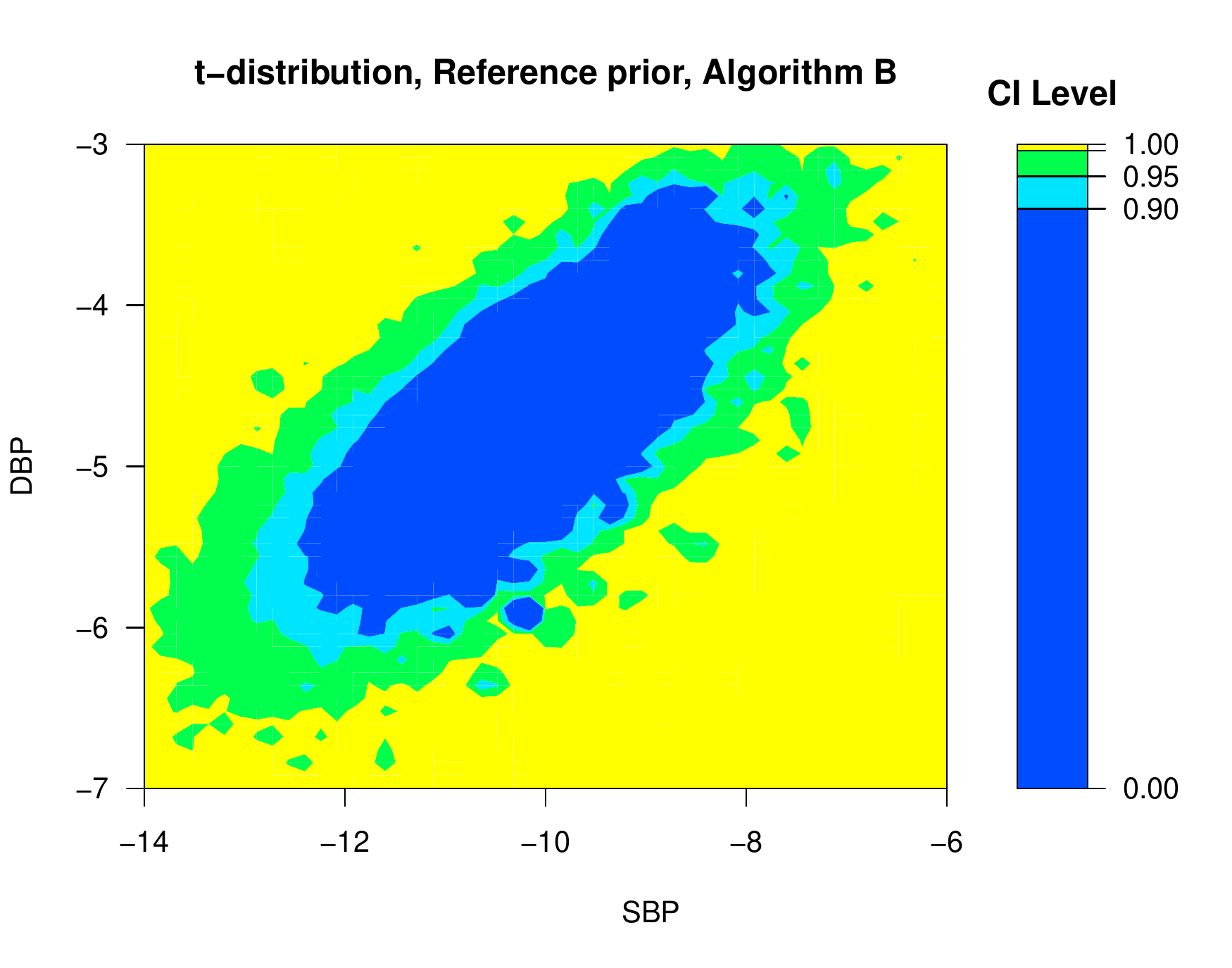}\\
\end{tabular}
 \caption{Credible sets for $\mu_1$ (SBP) and $\mu_2$ (DBP) obtained from the posterior distribution $\pi(\bmu|\bX)$ derived for the location parameters of the normal multivariate random effects model by employing the Berger and Bernardo reference prior and the Jeffrey prior and using data from Table \ref{tab:data}. The samples from the posterior distributions are drawn by Algorithm A and Algorithm B defined in Section \ref{sec:t}.}
\label{fig:t-emp-study}
 \end{figure}

The two-dimensional credible regions at significance levels 0.9 (dark blue), 0.95 (light blue), and 0.99 (green) for the elements of the mean vector $\bmu$ are depicted in Figure \ref{fig:norm-emp-study} for the normal multivariate random effects model and in Figure \ref{fig:t-emp-study} for the $t$ multivariate random effects model with $d=3$ degrees of freedom. The credible regions obtained under the assumption of the $t$-distribution are very similar, independently whether the Berger and Bernardo reference prior or the Jeffreys prior is employed, and the chosen algorithm to draw samples from the posterior distribution. That is not longer the case in Figure \ref{fig:norm-emp-study}, where the credible regions computed for the normal multivariate random effects model with Jeffreys prior and using Algorithm B appear to be slightly narrower. Finally, the credible intervals obtained under the assumption of the $t$-distribution are always wider reflecting the influence of heavy tails.

\section{Summary}\label{sec:sum}

Multivariate random effects model is one of the mostly used statistical tool in multivariate meta-analysis where the aim is to combine multiple values obtained in several studies into a single value. The parameters of the multivariate random effects model are usually estimated from the viewpoint of frequentist statistics, while several subjective Bayesian approaches based on the informative priors exist in the literature. Although both methods provide a good fit of the model to real data when the sample size is relatively large due to the asymptotic theorems of the frequentist statistics and the Bernstein-von-Mises theorem in Bayesian statistics, the results might be different when a sample of small size is present which is the case in the majority of meta-analyses. When the sample size is not large enough the asymptotic approximation might deviate considerable from the exact sample distribution of the estimated parameters or/and the influence of the chosen informative prior might have a significant impact on the posterior. Methods of the objective Bayesian statistics propose a solution to the challenges related to the insufficient sample size by endowing the models parameters with noninformative prior. In particular, the Berger and Bernardo reference prior is derived by maximizing the Shannon mutual information, i.e. by choosing the prior with the smallest impact on the posterior.

Flexible objective Bayesian procedures for the parameters of the multivariate random effects model are developed by employing two noninformative priors, the Berger and Bernardo reference prior and the Jeffreys prior. The analytical expressions of both the priors are obtained and the corresponding posteriors are derived. The results are established for a general class of multivariate random effects models which include the normal multivariate random effects model as a special case. Moreover, the propriety of the posteriors is proved under a weak condition, which requires that the sample size is larger than the dimension of the data generating-model only, independently of the specific class of the multivariate random effects model. Finally, the Metropolis-Hastings algorithm has been developed in the paper to draw samples from the posterior derived for the parameters of the model. Via simulations, it is shown that the considered numerical procedures lead to similar results in the case of the normal multivariate random effects model and the $t$ multivariate random effects model. In an empirical illustration based on data consisting of ten studies about the effectiveness of hypertension treatment for reducing blood pressure, a positive effect of the treatments on both the systolic blood pressure and diastolic blood pressure are found.

\section*{Acknowledgement}
This research was partially supported by National Institute of Standards and Technology (NIST) Exchange Visitor Program. The first author is grateful to the Statistical Engineering Division of National Institute of Standards and Technology (NIST) for providing an excellent and inspiring environment for research. This research is a part of the project \emph{Statistical Models and Data Reductions to Estimate Standard Atomic Weights and Isotopic Ratios for the Elements, and to Evaluate the Associated Uncertainties} (No. 2019-024-1-200), IUPAC (International Union of Pure and Applied Chemistry). Olha Bodnar also acknowledges valuable support from the internal grand (R\"{o}rlig resurs) of the \"{O}rebro University. Taras Bodnar was partially supported by the Swedish Research Council (VR) via the project \emph{Bayesian Analysis of Optimal Portfolios and Their Risk Measures}.

\section{Appendix}\label{sec:app}

In this section the proofs of theoretical results are given.

\begin{proof}[Proof of Theorem \ref{th1}:]
Under model (\ref{mult-rem}) the log-likelihood ignoring the constant term is given by
\begin{eqnarray}\label{likelihood_mult-rem}
L(\bmu,\bPsi; \bX )&=& -\frac{1}{2}\log(\text{det}(\bPsi\otimes\bI+\bU)) +\log \left(f\left(\text{vec}(\bX - \bmu \bi^\top)^\top (\bPsi \otimes \bI+\bU)^{-1}\text{vec}(\bX - \bmu \bi) \right)\right)\nonumber\\
&=&-\frac{1}{2}\sum_{i=1}^n\log(\text{det}(\bPsi+\bU_i)) + \log \left(f\left(\sum_{i=1}^{n}(\bx_i - \bmu)^\top (\bPsi + \bU_i)^{-1}(\bx_i - \bmu) \right)\right).
\end{eqnarray}
Hence,
\begin{equation*}\label{th1_pr_eq1}
\frac{\partial L(\bmu,\bPsi; \bX )}{\partial \bmu^\top}
=-2 \frac{f^\prime\left(\sum_{i=1}^{n}(\bx_i - \bmu)^\top (\bPsi + \bU_i)^{-1}(\bx_i - \bmu)\right)}
{f\left(\sum_{i=1}^{n}(\bx_i - \bmu)^\top (\bPsi + \bU_i)^{-1}(\bx_i - \bmu)\right)} \sum_{i=1}^{n}(\bx_i - \bmu)^\top(\bPsi + \bU_i)^{-1}.
\end{equation*}

Next, we compute the partial derivative of the log-likelihood function with respect to $\text{vech}(\bPsi)$, where $\text{vech}$ denote for the $vech$ operator with the following relation to $\text{vec}$ (see, \citet[p. 365]{Harville97})
\begin{eqnarray*}
\mathbf{vec}(\mathbf{\Psi})=\mathbf{G}_{p}\mathbf{vech}(\mathbf{\Psi}),
\end{eqnarray*}
where $\mathbf{G}_{p}$ is the duplication matrix (see, \citet[Section 3]{magnus2019matrix}).

From the properties of the differential of a determinant (see, \citet[Section 8]{magnus2019matrix}) we get
\begin{eqnarray*}
\mathbf{d} ~\text{det}(\bPsi + \bU_i)&=&\text{det}(\bPsi + \bU_i)(\text{vec}((\bPsi + \bU_i)^{-1}))^\top\mathbf{d}~ \text{vec}(\bPsi + \bU_i)\\
&=&\text{det}(\bPsi + \bU_i)(\text{vec}((\bPsi + \bU_i)^{-1}))^\top\mathbf{d}~ \text{vec}(\bPsi)\\
&=& \text{det}(\bPsi + \bU_i)(\text{vech}((\bPsi + \bU_i)^{-1}))^\top \mathbf{G}_{p}^\top\mathbf{G}_{p}\mathbf{d} ~\text{vech}(\bPsi).
\end{eqnarray*}
Thus,
\[\frac{\partial \text{det}(\bPsi + \bU_i)}{\partial \text{vech}(\bPsi)^\top}=\text{det}(\bPsi + \bU_i) \text{vech}((\bPsi + \bU_i)^{-1})^\top \mathbf{G}_{p}^\top\mathbf{G}_{p}\]
and
\begin{equation}\label{der_log_det}
\frac{\partial \log \text{det}(\bPsi + \bU_i)}{\partial \text{vech}(\bPsi)^\top}= \text{vech}((\bPsi + \bU_i)^{-1})^\top \mathbf{G}_{p}^\top\mathbf{G}_{p}.
\end{equation}

Similar, using the properties of the differential of inverse matrix (see, \citet[Section 8]{magnus2019matrix}), we obtain
\begin{eqnarray*}
\mathbf{d} ~(\bPsi + \bU_i)^{-1}&=&-(\bPsi + \bU_i)^{-1}\mathbf{d}~ (\bPsi + \bU_i)(\bPsi + \bU_i)^{-1}
=-(\bPsi + \bU_i)^{-1}\mathbf{d}~ \bPsi(\bPsi + \bU_i)^{-1}
\end{eqnarray*}
and, consequently,
\begin{eqnarray*}
\text{vec}(\mathbf{d} ~(\bPsi + \bU_i)^{-1})&=&\mathbf{d} ~\text{vec}((\bPsi + \bU_i)^{-1})\\
&=&-\text{vec} ((\bPsi + \bU_i)^{-1}\mathbf{d}~ \bPsi(\bPsi + \bU_i)^{-1})\\
&=&-\left((\bPsi + \bU_i)^{-1} \otimes (\bPsi + \bU_i)^{-1}\right)\mathbf{d}\text{vec}(\bPsi)\\
&=&-\left((\bPsi + \bU_i)^{-1} \otimes (\bPsi + \bU_i)^{-1}\right)\mathbf{G}_{p}\mathbf{d} ~\text{vech}(\bPsi).
\end{eqnarray*}
Hence,
\begin{equation}\label{der_inv}
\frac{\partial \text{vec}(\bPsi + \bU_i)^{-1}}{\partial \text{vech}(\bPsi)^\top}=- \left((\bPsi + \bU_i)^{-1} \otimes (\bPsi + \bU_i)^{-1}\right)\mathbf{G}_{p}.
\end{equation}

The application of \eqref{der_log_det} and \eqref{der_inv} leads to
\begin{eqnarray*}
&&\frac{\partial L(\bmu,\bPsi; \bX )}{\partial \text{vech}(\bPsi)^\top}
=-\frac{1}{2}\sum_{i=1}^n \frac{\partial\log(\text{det}(\bPsi+\bU_i))}{\partial \text{vech}(\bPsi)^\top}\\
&+& \frac{\partial}{\partial \text{vech}(\bPsi)^\top}\log \left(f\left(\sum_{i=1}^{n}(\bx_i - \bmu)^\top (\bPsi + \bU_i)^{-1}(\bx_i - \bmu) \right)\right)\\
&=& -\frac{1}{2}\sum_{i=1}^n \text{vech}((\bPsi + \bU_i)^{-1})^\top \mathbf{G}_{p}^\top\mathbf{G}_{p}
+ \frac{f^\prime\left(\sum_{i=1}^{n}(\bx_i - \bmu)^\top (\bPsi + \bU_i)^{-1}(\bx_i - \bmu)\right)}
{f\left(\sum_{i=1}^{n}(\bx_i - \bmu)^\top (\bPsi + \bU_i)^{-1}(\bx_i - \bmu)\right)} \\
&\times& \sum_{i=1}^{n}\left((\bx_i - \bmu)^\top \otimes (\bx_i - \bmu)^\top\right) \frac{\partial \text{vec}(\bPsi + \bU_i)^{-1}}{\partial \text{vech}(\bPsi)^\top}\\
&=& -\frac{1}{2}\sum_{i=1}^n\text{vec}((\bPsi + \bU_i)^{-1})^\top \mathbf{G}_{p}
- \frac{f^\prime\left(\sum_{i=1}^{n}(\bx_i - \bmu)^\top (\bPsi + \bU_i)^{-1}(\bx_i - \bmu)\right)}
{f\left(\sum_{i=1}^{n}(\bx_i - \bmu)^\top (\bPsi + \bU_i)^{-1}(\bx_i - \bmu)\right)} \\
&\times& \sum_{i=1}^{n}\left((\bx_i - \bmu)^\top \otimes (\bx_i - \bmu)^\top\right) \left((\bPsi + \bU_i)^{-1} \otimes (\bPsi + \bU_i)^{-1}\right)\mathbf{G}_{p}
\,,
\end{eqnarray*}

The first block of the Fisher information matrix is given by
\begin{eqnarray}\label{F11_app_eq1}
\bF_{11}&=&\E\left[\left(\frac{\partial L(\bmu,\bPsi; \bX )}{\partial \bmu^\top}\right)^\top\frac{\partial L(\bmu,\bPsi; \bX )}{\partial \bmu^\top}\right]
=4\sum_{i=1}^{n}\sum_{j=1}^{n}(\bPsi + \bU_i)^{-1} \bH_{ij} (\bPsi + \bU_j)^{-1}
\end{eqnarray}
with
\begin{eqnarray*}
\bH_{ij}&=& \left(\prod_{k=1}^{n}\sqrt{\text{det}(\bPsi + \bU_k)}^{-1/2}\right)\int_{\R^{np}} (\bx_i - \bmu)(\bx_j - \bmu)^\top \\
&\times&
\left(\frac{f^\prime\left(\sum_{k=1}^{n}(\bx_k - \bmu)^\top (\bPsi + \bU_k)^{-1}(\bx_k - \bmu)\right)}
{f\left(\sum_{k=1}^{n}(\bx_k - \bmu)^\top (\bPsi + \bU_k)^{-1}(\bx_k - \bmu)\right)} \right)^2 f\left(\sum_{k=1}^{n}(\bx_k - \bmu)^\top (\bPsi + \bU_k)^{-1}(\bx_k - \bmu)\right) \mathbf{d}\bX\\
&=& \int_{\R^{np}} (\bPsi + \bU_i)^{1/2}\bz_i \bz_j^\top(\bPsi + \bU_j)^{1/2}
\left(\frac{f^\prime\left(\sum_{k=1}^{n}\bz_k^\top\bz_k\right)}
{f\left(\sum_{k=1}^{n}\bz_k^\top\bz_k\right)} \right)^2 f\left(\sum_{k=1}^{n}\bz_k^\top\bz_k\right) \mathbf{d}\bZ,
\end{eqnarray*}
where the last equality follows from the transformation $\bx_k=\bmu+(\bPsi + \bU_k)^{1/2}\bz_k$ for $k=1,...,n$ with the Jacobian equal to $ \left(\prod_{k=1}^{n}\sqrt{\text{det}(\bPsi + \bU_k)}^{1/2}\right)$ and $\bZ=(\bz_1,\ldots,\bz_n)$. Moreover, we get that
\begin{equation*}
\text{vec}(\bZ)=(\bPsi \otimes \bI+\bU)^{-1/2}\text{vec}(\bX - \bmu \bi)
\end{equation*}
and, consequently, $\bZ\sim E_{p,n}(\mathbf{O}_{p,n},\bI_{p\times n},f)$ with $\mathbf{O}_{p,n}$ $p \times n$ zero matrix (see, Theorem 2.13 in \citet{GuptaVargaBodnar}).

Since
\[\bz_j^\top(\bPsi + \bU_j)^{1/2}
\left(\frac{f^\prime\left(\sum_{k=1}^{n}\bz_k^\top\bz_k\right)}
{f\left(\sum_{k=1}^{n}\bz_k^\top\bz_k\right)} \right)^2 f\left(\sum_{k=1}^{n}\bz_k^\top\bz_k\right)\]
is an odd function on a symmetric region, we get that $\bH_{ij}=\mathbf{O}_{p,p}$. Furthermore, using that $\bz_i/\sqrt{\text{vec}(\bZ)^\top\text{vec}(\bZ)}$ and $\text{vec}(\bZ)^\top\text{vec}(\bZ)$ with $\text{vec}(\bZ)^\top\text{vec}(\bZ)=\sum_{k=1}^{n}\bz_k^\top\bz_k$ are independent (see, \citet[Theorem 2.15]{GuptaVargaBodnar}), we obtain
\begin{eqnarray}\label{Hii}
\bH_{ii}&=& J_1 (\bPsi + \bU_i)^{1/2}\E\left(\frac{\bz_i\bz_i^\top}{\text{vec}(\bZ)^\top\text{vec}(\bZ)}\right)(\bPsi + \bU_j)^{1/2},
\end{eqnarray}
where
\begin{eqnarray}\label{J1_app}
J_1&=&\E\left(\text{vec}(\bZ)^\top\text{vec}(\bZ)\left(\frac{f^\prime\left(\text{vec}(\bZ)^\top\text{vec}(\bZ)\right)}
{f\left(\text{vec}(\bZ)^\top\text{vec}(\bZ)\right)} \right)^2 \right).
\end{eqnarray}

To this end, we note that the distribution of $\bz_i/\sqrt{\text{vec}(\bZ)^\top\text{vec}(\bZ)}$ does not depend on the type of elliptical distribution, i.e., on density generator $f(.)$, and it is the same as in case of the normal distribution. Let $\bZ_N \sim \mathcal{N}_{p,n}(\mathbf{O}_{p,n},\bI_{p\times n})$ (standard matrix-variate normal distribution) and let $\bz_{i,N}$ denote its $i$th column. Then, it holds that
\begin{eqnarray}\label{Ezz}
&&\E\left(\frac{\bz_i\bz_i^\top}{\text{vec}(\bZ)^\top\text{vec}(\bZ)}\right)=\E\left(\frac{\bz_{i,N}\bz_{i,N}^\top}{\text{vec}(\bZ_N)^\top\text{vec}(\bZ_N)}\right)
\frac{\E(\text{vec}(\bZ_N)^\top\text{vec}(\bZ_N))}{\E(\text{vec}(\bZ_N)^\top\text{vec}(\bZ_N))} \nonumber\\
&=&\frac{\E\left(\bz_{i,N}\bz_{i,N}^\top\right)}{\E(\text{vec}(\bZ_N)^\top\text{vec}(\bZ_N))}= \frac{1}{pn} \bI_p,
\end{eqnarray}
where we used that $\text{vec}(\bZ_N)^\top\text{vec}(\bZ_N)\sim \chi^2_{pn}$ ($\chi^2$-distribution with $pn$ degrees of freedom).

Summarizing \eqref{F11_app_eq1}, \eqref{Hii}, and \eqref{J1_app}, we get
\begin{equation*}
\bF_{11}=\frac{4J_1}{pn}\sum_{i=1}^{n}(\bPsi + \bU_i)^{-1}.
\end{equation*}

Next, we compute the nondiagonal block $\bF_{21}$ of the Fisher information matrix. The transformation $\bx_k=\bmu+(\bPsi + \bU_k)^{1/2}\bz_k$ for $k=1,...,n$ yields
\begin{eqnarray*}\label{F21_app_eq1}
\bF_{21}&=&\E\left[\left(\frac{\partial L(\bmu,\bPsi; \bX )}{\partial \bmu^\top}\right)^\top\frac{\partial L(\bmu,\bPsi; \bX )}{\partial \text{vech}(\bPsi)^\top}\right]
=\sum_{i=1}^{n}\sum_{j=1}^{n}(\bPsi + \bU_i)^{-1/2} \left(\int_{\R^{pn}} \mathbf{M}_{ij}(\bZ)\mathbf{d} \bZ\right) \mathbf{G}_p,
\end{eqnarray*}
where
\begin{eqnarray*}
\mathbf{M}_{ij}(\bZ)&=&\frac{f^\prime\left(\sum_{k=1}^{n}\bz_k^\top\bz_k\right)}{f\left(\sum_{k=1}^{n}\bz_k^\top\bz_k\right)}\bz_i
\Bigg[\text{vec}((\bPsi + \bU_j)^{-1})^\top \\
&+&2\frac{f^\prime\left(\sum_{k=1}^{n}\bz_k^\top\bz_k\right)}{f\left(\sum_{k=1}^{n}\bz_k^\top\bz_k\right)}
\left(\bz_j^\top \otimes \bz_j^\top\right) \left((\bPsi + \bU_j)^{-1/2} \otimes (\bPsi + \bU_j)^{-1/2}\right)
\Bigg]
\end{eqnarray*}
which is an odd function on a symmetric region. Hence,
\[\int_{\R^{pn}} \mathbf{M}_{ij}(\bZ)\mathbf{d} \bZ=\mathbf{O} \quad \text{for all} \quad i,j \in \{1,...,n\},\]
and, consequently, $\bF_{21}=\mathbf{O}$.

Similarly, using the transformation $\bx_k=\bmu+(\bPsi + \bU_k)^{1/2}\bz_k$ for $k=1,...,n$, we get
\begin{eqnarray*}
\bF_{22}&=&\E\left[\left(\frac{\partial L(\bmu,\bPsi; \bX )}{\partial \text{vech}(\bPsi)^\top}\right)^\top\frac{\partial L(\bmu,\bPsi; \bX )}{\partial \text{vech}(\bPsi)^\top}\right]\\
&=&\E\Bigg\{\mathbf{G}_p^\top\Bigg[\sum_{i=1}^{n}\sum_{j=1}^{n} \frac{1}{4} \text{vec}((\bPsi + \bU_i)^{-1})\text{vec}((\bPsi + \bU_j)^{-1})^\top \\
&+&\frac{1}{2}\frac{f^\prime\left(\sum_{k=1}^{n}\bz_k^\top\bz_k\right)}{f\left(\sum_{k=1}^{n}\bz_k^\top\bz_k\right)}
\left((\bPsi + \bU_i)^{-1/2}\otimes (\bPsi + \bU_i)^{-1/2}\right)\left(\bz_i \otimes \bz_i\right)\text{vec}((\bPsi + \bU_j)^{-1})^\top \\
&+&\frac{1}{2}\frac{f^\prime\left(\sum_{k=1}^{n}\bz_k^\top\bz_k\right)}{f\left(\sum_{k=1}^{n}\bz_k^\top\bz_k\right)}
 \text{vec}((\bPsi + \bU_i)^{-1})\left(\bz_j^\top \otimes \bz_j^\top\right) \left((\bPsi + \bU_j)^{-1/2} \otimes (\bPsi + \bU_j)^{-1/2}\right)
 \\
&+&\left(\frac{f^\prime\left(\sum_{k=1}^{n}\bz_k^\top\bz_k\right)}{f\left(\sum_{k=1}^{n}\bz_k^\top\bz_k\right)}\right)^2
\left((\bPsi + \bU_i)^{-1/2} \otimes (\bPsi + \bU_i)^{-1/2}\right)\\
&\times&
\left(\bz_i \otimes \bz_i\right) \left(\bz_j^\top \otimes \bz_j^\top\right) \left((\bPsi + \bU_j)^{-1/2} \otimes (\bPsi + \bU_j)^{-1/2}\right)
\mathbf{G}_p\Bigg]\Bigg\}\\
&=&\mathbf{G}_p^\top\Bigg[\sum_{i=1}^{n}\sum_{j=1}^{n} \frac{1}{4} \text{vec}((\bPsi + \bU_i)^{-1})\text{vec}((\bPsi + \bU_j)^{-1})^\top \\
&+&\frac{1}{2}\left((\bPsi + \bU_i)^{-1/2}\otimes (\bPsi + \bU_i)^{-1/2}\right)\mathbf{D}_{i}\text{vec}((\bPsi + \bU_j)^{-1})^\top \\
&+&\frac{1}{2} \text{vec}((\bPsi + \bU_i)^{-1})\mathbf{D}_{j}^\top \left((\bPsi + \bU_j)^{-1/2} \otimes (\bPsi + \bU_j)^{-1/2}\right) \\
&+&\left((\bPsi + \bU_i)^{-1/2} \otimes (\bPsi + \bU_i)^{-1/2}\right)\mathbf{D}_{ij} \left((\bPsi + \bU_j)^{-1/2} \otimes (\bPsi + \bU_j)^{-1/2}\right)
\mathbf{G}_p\Bigg]
\end{eqnarray*}
with
\begin{equation}\label{Di}
\mathbf{D}_{i}=\E\left(\frac{f^\prime\left(\text{vec}(\bZ)^\top\text{vec}(\bZ)\right)}{f\left(\text{vec}(\bZ)^\top\text{vec}(\bZ)\right)}\left(\bz_i \otimes \bz_i\right)\right)
=J\E\left(\frac{\left(\bz_i \otimes \bz_i\right)}{\text{vec}(\bZ)^\top\text{vec}(\bZ)}\right)
\end{equation}
and
\begin{equation}\label{Dij}
\mathbf{D}_{ij}=\E\left(\left(\frac{f^\prime\left(\text{vec}(\bZ)^\top\text{vec}(\bZ)\right)}{f\left(\text{vec}(\bZ)^\top\text{vec}(\bZ)\right)}\right)^2 \left(\bz_i\otimes \bz_i\right) \left(\bz_j^\top \otimes \bz_j^\top\right) \right)
=J_2\E\left(\frac{\left(\bz_i \otimes \bz_i\right) \left(\bz_j^\top \otimes \bz_j^\top\right)}{(\text{vec}(\bZ)^\top\text{vec}(\bZ))^2} \right)
\end{equation}
where we use that $\bZ/\sqrt{\text{vec}(\bZ)^\top\text{vec}(\bZ)}$ and $\text{vec}(\bZ)^\top\text{vec}(\bZ)$ are independent (see, \citet[Theorem 2.15]{GuptaVargaBodnar}), $\sum_{k=1}^{n}\bz_k^\top\bz_k=\text{vec}(\bZ)^\top\text{vec}(\bZ)$, and define
\begin{eqnarray}\label{J2_app}
J_2&=&\E\left((\text{vec}(\bZ)^\top\text{vec}(\bZ))^2\left(\frac{f^\prime\left(\text{vec}(\bZ)^\top\text{vec}(\bZ)\right)}
{f\left(\text{vec}(\bZ)^\top\text{vec}(\bZ)\right)} \right)^2 \right).
\end{eqnarray}
and
\begin{eqnarray*}
J&=&\E\left(\text{vec}(\bZ)^\top\text{vec}(\bZ)\frac{f^\prime\left(\text{vec}(\bZ)^\top\text{vec}(\bZ)\right)}
{f\left(\text{vec}(\bZ)^\top\text{vec}(\bZ)\right)} \right).
\end{eqnarray*}

Let $R^2=\text{vec}(\bZ)^\top\text{vec}(\bZ)$. Then the density of $R^2$ is given by $f_{R^2}(r)=r^{pn/2-1}f(r)$ (cf., Theorem 2.16 in \cite{GuptaVargaBodnar})
 and
\begin{eqnarray*}
J&=&E\left(R^2 \frac{f^\prime(R^2)}{f(R^2)}\right)=\int_0^\infty r \frac{f^\prime(r)}{f(r)} r^{pn/2-1}f(r) \mbox{d}r =\int_0^\infty r^{pn/2} f^\prime(r) \mbox{d}r\\
&=&r^{pn/2} f(r)\Bigg|_0^\infty-\frac{np}{2}\int_0^\infty r^{pn/2-1} f(r) \mbox{d}r=-\frac{pn}{2} \,,
\end{eqnarray*}
where we use that $\int_0^\infty r^{n/2-1} f(r) \mbox{d}r=1$ because $r^{n/2-1} f(r)$ is the density of $R^2$ which also implies that $r^{n/2-1} f(r)=o(r^{-1})$ as $r \rightarrow \infty$, i.e. $r^{n/2} f(r)=o(1)$ as $r \rightarrow \infty$.

Moreover, since the distribution of $\bz_i/\sqrt{\text{vec}(\bZ)^\top\text{vec}(\bZ)}$ does not depend on the type of elliptical distribution, we get with $\bZ_N \sim \mathcal{N}_{p,n}(\mathbf{O}_{p,n},\bI_{p\times n})$ that
\begin{eqnarray*}
&&\E\left(\frac{\left(\bz_i \otimes \bz_i\right)}{\text{vec}(\bZ)^\top\text{vec}(\bZ)}\right)
=\E\left(\frac{\left(\bz_{i,N} \otimes \bz_{i,N}\right)}{\text{vec}(\bZ_N)^\top\text{vec}(\bZ_N)}\right)\\
&=&\E\left(\frac{\left(\bz_{i,N} \otimes \bz_{i,N}\right)}{\text{vec}(\bZ_N)^\top\text{vec}(\bZ_N)}\right)
\frac{\E\left(\text{vec}(\bZ_N)^\top\text{vec}(\bZ_N)\right)}{\E\left(\text{vec}(\bZ_N)^\top\text{vec}(\bZ_N)\right)}
=\E\left(\bz_{i,N} \otimes \bz_{i,N}\right) \frac{1}{pn}= \frac{1}{pn}\text{vec}(\bI_p),
\end{eqnarray*}
where the last equality follows from the results in \citet[p. 67]{ghazal2000second}

Similarly for $i \neq j \in\{1,...,n\}$ we get
\begin{eqnarray*}
&&\E\left(\frac{\left(\bz_i \otimes \bz_i\right) \left(\bz_j^\top \otimes \bz_j^\top\right)}{(\text{vec}(\bZ)^\top\text{vec}(\bZ))^2} \right)\\
&=&\E\left(\left(\bz_{i,N} \otimes \bz_{i,N}\right) \left(\bz_{j,N}^\top \otimes \bz_{j,N}^\top\right)\right)\frac{1}{\E\left((\text{vec}(\bZ_N)^\top\text{vec}(\bZ_N))^2\right)}\\
&=&\E\left(\bz_{i,N} \otimes \bz_{i,N}\right) \E\left(\bz_{j,N}^\top \otimes \bz_{j,N}^\top\right)\frac{1}{2pn+p^2n^2}\\
&=&\frac{1}{2pn+p^2n^2} \text{vec}(\bI_p)\text{vec}(\bI_p)^\top.
\end{eqnarray*}

Finally, for $i=j \in \{1,...,n\}$ we obtain
\begin{eqnarray*}
&&\E\left(\frac{\left(\bz_i \otimes \bz_i\right) \left(\bz_i^\top \otimes \bz_k^\top\right)}{(\text{vec}(\bZ)^\top\text{vec}(\bZ))^2} \right)\\
&=&\E\left(\left(\bz_{i,N} \otimes \bz_{i,N}\right) \left(\bz_{i,N}^\top \otimes \bz_{i,N}^\top\right)\right)\frac{1}{\E\left((\text{vec}(\bZ_N)^\top\text{vec}(\bZ_N))^2\right)}\\
&=& (\bI_{p^2}+\mathbf{K}_p+\text{vec}(\bI_p)\text{vec}(\bI_p)^\top) \frac{1}{2pn+p^2n^2},
\end{eqnarray*}
where the last equality follows from Theorem 4.1 in \citet{magnus1979commutation} and $\mathbf{K}_p$ is the commutation matrix.

The properties of the $vec$ operator and the properties of the duplication and commutation matrices, namely (see, \citet[Theorem 3.12 and Theorem 3.13]{magnus2019matrix})
\[\frac{1}{2}((\bI_{p^2}+\mathbf{K}_p)(\mathbf{A} \otimes \mathbf{A}) \mathbf{G}_p=(\mathbf{A} \otimes \mathbf{A}) \mathbf{G}_p \]
 yield
\begin{eqnarray*}
&&\left((\bPsi + \bU_i)^{-1/2}\otimes (\bPsi + \bU_i)^{-1/2}\right)\text{vec}(\bI_p)\\
&=&\text{vec}\left((\bPsi + \bU_i)^{-1/2}(\bPsi + \bU_i)^{-1/2}\right)=\text{vec}\left((\bPsi + \bU_i)^{-1}\right)
\end{eqnarray*}
and
\begin{eqnarray*}
&&\left((\bPsi + \bU_i)^{-1/2} \otimes (\bPsi + \bU_i)^{-1/2}\right)(\bI_{p^2}+K_p)\left((\bPsi + \bU_i)^{-1/2} \otimes (\bPsi + \bU_i)^{-1/2}\right)\mathbf{G}_p\\
&=&2\left((\bPsi + \bU_i)^{-1} \otimes (\bPsi + \bU_i)^{-1}\right)\mathbf{G}_p.
\end{eqnarray*}

Hence,
\begin{eqnarray*}
\bF_{22}
&=&\mathbf{G}_p^\top\Bigg[ \left( \frac{J_2}{2pn+p^2n^2} -\frac{1}{4}\right) \text{vec}\left(\sum_{i=1}^{n}(\bPsi + \bU_i)^{-1}\right)\text{vec}\left(\sum_{j=1}^{n}(\bPsi + \bU_j)^{-1}\right)^\top \\
&+& \frac{2J_2}{2pn+p^2n^2}\sum_{i=1}^{n}\left((\bPsi + \bU_i)^{-1} \otimes (\bPsi + \bU_i)^{-1}\right)
\Bigg]\mathbf{G}_p,
\end{eqnarray*}
which complete the proof of the theorem.
\end{proof}

\begin{proof}[Proof of Lemma~\ref{lem1}:]
First, we show that $(\bA+\bB)\otimes (\bA +\bB)-\bA\otimes \bA \ge \mathbf{0}$. For any vector $\bc=\text{vec}(\bC)$ we get
\begin{eqnarray*}
&&\bc^\top \left( (\bA+\bB)\otimes (\bA +\bB)-\bA\otimes \bA \right) \bc\\
 &=&\text{vec}(\bC)^\top \left( (\bA+\bB)\otimes (\bA +\bB)-\bA\otimes \bA \right) \text{vec}(\bC)\\
&=& \text{tr}(\bC^\top (\bA+\bB) \bC (\bA+\bB))- tr(\bC^\top \bA \bC \bA)\\
&=& \text{tr}(\bC^\top (\bA+\bB) \bC (\bA+\bB))- tr(\bC^\top (\bA+\bB) \bC \bA)+tr(\bC^\top (\bA+\bB) \bC \bA) tr(\bC^\top \bA \bC \bA)\\
&=& \text{tr}(\bC^\top (\bA+\bB) \bC \bB)+ tr(\bC^\top \bB \bC \bA)\\
&=& \text{tr}( (\bA+\bB)^{1/2} \bC \bB\bC^\top(\bA+\bB)^{1/2})+ tr(\bA^{1/2}\bC^\top \bB \bC \bA^{1/2})\ge 0.
\end{eqnarray*}
Since $\bA$ and $\bB$ are symmetric, we get that $(\bA+\bB)\otimes (\bA +\bB)$ and $\bA\otimes \bA$ are symmetric. Hence, the application of Theorem 18.3.4 in \citet{Harville97} leads to the second statement of the lemma.

\end{proof}
{\footnotesize
\bibliography{Bibfile2019-08-19}
}
\end{document}